%% file: root.tex
\pdfoutput=1

\documentclass[format=acmsmall, review=false, screen=true]{acmart}

\usepackage{booktabs} 

\usepackage{geometry}
\usepackage{amsmath}
\usepackage{amssymb}
\usepackage{bussproofs}
\usepackage{graphicx}
\usepackage{comment}
\usepackage{multirow}

\usepackage{dcolumn,booktabs}
\newcolumntype{d}[1]{D{.}{.}{#1}}

\usepackage{changes}
\definechangesauthor[name={Alwen Tiu}, color=orange]{at}

\usepackage[colorinlistoftodos,textsize=scriptsize]{todonotes}

\usepackage[ruled]{algorithm2e} 
\renewcommand{\algorithmcfname}{ALGORITHM}
\SetAlFnt{\small}
\SetAlCapFnt{\small}
\SetAlCapNameFnt{\small}
\SetAlCapHSkip{0pt}
\IncMargin{-\parindent}

\usepackage[all]{xy} 
\SelectTips{cm}{}                 

\newcommand{\limp}[0]{\ensuremath{\rightarrow}}

\newcommand{\mimp}[0]{\ensuremath{{\;-\!*\;}}}
\newcommand{\munit}{\ensuremath{\top^*}}

\newcommand{\simp}[0]{\ensuremath{\triangleright}}

\newcommand{\mand}[0]{\ensuremath{*}}
\newcommand{\lsbbi}[0]{\ensuremath{\mathrm{LS_{BBI}}}}
\newcommand{\fvlsbbi}[0]{\ensuremath{\mathrm{FVLS_{BBI}}}}

\newcommand{\psl}[0]{\mbox{PASL}}
\newcommand{\lspslh}[0]{\ensuremath{\mathrm{LS_{PASL}}}}

\newcommand{\lsg}[0]{\ensuremath{\mathrm{LS_G}}}
\newcommand{\ls}[0]{\mbox{BBI}^-}

\newcommand{\mymodel}[4]{\ensuremath{(#1, #2, #3, #4)}}
\newcommand{\mylmodel}[5]{\ensuremath{(#1, #2, #3, #4, #5)}}
\newcommand{\myseq}[3]{\ensuremath{#1; #2 \vdash #3}}

\newcommand{\pfun}{\rightharpoonup}

\def\Fcal{\mathcal{F}}
\def\Gcal{\mathcal{G}}
\def\Mcal{\mathcal{M}}
\def\Ncal{\mathcal{N}}
\def\Lcal{\mathcal{L}}
\def\Pcal{\mathcal{P}}
\def\Scal{\mathcal{S}}

\def\Rcal{\mathcal{R}}
\def\Frcal{\mathcal{F}r}

\def\val{\nu}

\def\EMbb{\mathbb{E}\mathbb{M}}

\def\GSbb{\mathbb{G}\mathbb{S}}

\def\LVar{\mathit{LVar}}

\newcommand{\exfml}[1]{Ex#1}

\acmJournal{TOCL}
\acmVolume{0}
\acmNumber{0}
\acmArticle{00}
\acmYear{0000}
\acmMonth{0}
\copyrightyear{0000}

\setcopyright{acmlicensed}

\acmDOI{0000001.0000001}

\begin{document}
\title{Modular Labelled Sequent Calculi for Abstract Separation Logics}  

\author{Zh\'e H\'ou}
\affiliation{
  \institution{Griffith University}
  \department{Institute for Integrated and Intelligent Systems}
  \city{Brisbane}
  \state{QLD}
  \postcode{4111}
  \country{Australia}
}

\author{Ranald Clouston}
\affiliation{
  \institution{Aarhus University}
  \department{Department of Computer Science}
  \city{Aarhus}
  \postcode{8200}
  \country{Denmark}
}

\author{Rajeev Gor\'e}
\affiliation{
  \institution{The Australian National University}
  \department{Research School of Computer Science}
  \city{Canberra}
  \state{ACT}
  \postcode{2601}
  \country{Australia}
}

\author{Alwen Tiu}
\affiliation{
  \institution{The Australian National University}
  \department{Research School of Computer Science}
  \city{Canebrra}
  \postcode{2601}
  \country{Australia}
}


\begin{abstract}
  Abstract separation logics are a family of extensions of Hoare logic
  for reasoning about programs that manipulate resources such as
  memory locations.  These logics are
  ``abstract'' because they are independent of any particular concrete
  resource model. Their assertion languages, called propositional
  abstract separation logics (PASLs), extend the logic of (Boolean)
  Bunched Implications (BBI) in various ways. In particular, these
  logics contain the connectives $\mand$ and $\mimp$, denoting the
  composition and extension of resources respectively.

  This added expressive power comes at a price since the
  resulting logics are all undecidable. Given their wide
  applicability, even a semi-decision procedure for these logics is
  desirable.  Although several PASLs and their relationships
  with BBI are discussed in the literature, the proof theory of, and
  automated reasoning for, these logics were open problems solved by
  the conference version of this paper, which developed a modular
  proof theory
  for various PASLs using cut-free labelled sequent calculi.
  This paper non-trivially improves upon this previous work
  by giving a general framework of calculi on which any new axiom in
  the logic satisfying a certain form corresponds to an inference rule
  in our framework, and the completeness proof is generalised to
  consider such axioms.

  Our base calculus handles Calcagno et al.'s original logic of
  separation algebras by adding sound rules for partial-determinism
  and cancellativity, while preserving cut-elimination. We then show
  that many important properties in separation logic, such as
  indivisible unit, disjointness, splittability, and cross-split, can
  be expressed in our general axiom form. Thus our framework
  offers inference rules and completeness for these properties for
  free. Finally, we show how our calculi reduce to
  calculi with global label substitutions,
  enabling more efficient implementation.
\end{abstract}

%
%


%
%





\maketitle


\input{intro}
\input{psl}
\input{counter_model_psl}
\input{extension_psl}
\input{equality_subs}
\input{examples}
\input{experiment}
\input{applications}
\input{rel_work}

\paragraph{Acknowledgements} We thank the anonymous reviewers of this paper for their helpful comments.
Clouston was supported by a research grant (12386) from Villum Fonden.
Tiu and H\'ou received support from the National Research Foundation of Singapore under its National Cybersecurity R\&D Program (Award No. NRF2014NCR-NCR001-30) and administered by the National Cybersecurity R\&D Directorate.

\bibliographystyle{ACM-Reference-Format}
\bibliography{root}

\end{document}

%% file: intro.tex
\section{Introduction}
\label{sec:intro}

Reynolds's Separation logic (SL)~\cite{reynolds2002} is an extension of 
Hoare logic for reasoning about programs that explicitly mutate memory. Its
assertion logic, also called separation logic, extends the usual 
(additive) connectives 
for conjunction $\land$, disjunction $\lor$,
implication $\to$, and the (additive) verum constant $\top$,
with the multiplicative connectives
\emph{separating conjunction} $\mand$, its unit $\munit$ (denoted by $emp$ in some literature), and
\emph{separating implication} $\mimp$, also called \emph{magic wand}, from the
logic of Bunched Implications (BI)~\cite{OHearnPym1999}.
Moreover, the assertion language introduces 
the \emph{points-to} predicate $E \mapsto E'$
on expressions,
along with the usual quantifiers and predicates of first-order logic with equality
and arithmetic.
The additive connectives
may be either intuitionistic, as for BI, or classical, as for
the logic of \emph{Boolean} Bunched Implications (BBI). Classical
additives are more expressive as they support reasoning about
non-monotonic commands such as memory de-allocation, and assertions
such as ``the heap is empty''~\cite{IshtiaqOHearn01}. In this paper we
consider classical additives only.

The concrete memory model for SL is given in terms of heaps,
where a \emph{heap} is a finite partial function from addresses to
values. A heap satisfies $P\mand Q$ iff it can be partitioned into
heaps satisfying $P$ and $Q$ respectively; it satisfies $\munit$ iff
it is empty; it satisfies $P\mimp Q$ iff any extension with a heap
that satisfies $P$ must then satisfy $Q$; and it satisfies $E\mapsto
E'$ iff it is a singleton map sending the address specified by the
expression $E$ to the value specified by the expression $E'$.
While the $\mapsto$ predicate refers to the \emph{content} of heaps, the BI
connectives refer only to their \emph{structure}.
Some basic spatial properties of heaps include the following:
\begin{description}
\item[\underline{Empty heap}] There is a unique empty heap $\epsilon$;
\item[\underline{Identity}] Combining heap $h$ with the empty heap $\epsilon$
  gives the original heap $h$;
\item[\underline{Commutativity}] Combining heap $h_1$ with heap $h_2$ is the
  same as combining $h_2$ with $h_1$;
\item[\underline{Associativity}] Combining heap $h_1$ with heap $h_2$ and then
  combining the result with heap $h_3$ is the same as combining heap
  $h_1$ with the combination of heaps $h_2$ and $h_3$.
\end{description}
These conditions define a \emph{non-deterministic monoid}: giving
algebraic models for BBI~\cite{GalmicheLarcheyWendling2006}.

The idea of separation logic has proved fruitful for a range of memory (or, more
generally, resource) models, some quite different from the original heap model.
In this paper we will present examples drawn
from~\cite{Boyland:Checking,parkinson2005,Bornat:Permission,calcagno2007,dockins2009,Villard:Proving,Jensen:High},
but this list is far from exhaustive.
Each such model has its own notion of separation and sharing of resources, and
hence may formally give rise to a new logic with respect to the BI connectives, let alone
any special-purpose predicates which might be added to the logic. As new
variations of separation logic are introduced, their relation to prior logics is seldom
developed formally, and so new
metatheory and tool support must be substantially reconstructed for each case.
This has led to a subgenre of papers highlighting the need for organisation and
generalisation across these
logics~\cite{Biering:BI,calcagno2007,Parkinson:Next,Jensen:Techniques}.

In this paper we take as a starting point \emph{Abstract Separation
Logic}~\cite{calcagno2007}, which is intended to generalise the logics of many
concrete models. In particular
we set quantifiers aside to work with Propositional Abstract Separation Logic ($\psl$).
This logic is defined via the
abstract semantics of partial cancellative monoids, or
\emph{separation algebras}, which are non-deterministic
monoids restricted by: 
\begin{description}
\item[\underline{Partial-determinism}] The combination of heap $h_1$ with heap
  $h_2$ is either undefined, or a unique heap;
\item[\underline{Cancellativity}] If combining 
  $h_1$ and $h_2$ gives $h_3$ and combining heap $h_1$ and  
  $h_4$ also gives $h_3$, then $h_2 = h_4$.
\end{description}
Semantics in this style are reminiscent of the ternary algebraic semantics often used in
connection with substructural logics~\cite{Anderson:Entailment}, an observation we
exploit in this paper.
Separation algebras allow
interpretation of $\mand$, $\munit$ and $\mimp$, although the latter
is not considered by~\cite{calcagno2007}.
The points-to ($\mapsto$) predicate is not
a first class citizen of $\psl$; it may be introduced as a predicate only if an
appropriate concrete separation algebra is fixed.
$\psl$ is appropriate to
reasoning about the \emph{structure} of memory, but not its \emph{content}.

Precondition strengthening and postcondition weakening in Hoare-style logics require
reasoning in the assertion logic, 
but proof search and structural proof theory for
$\psl$ have received little attention until recently. It
is known that the added expressive power of the multiplicative
connectives comes at a price, yielding a logic that is in general
undecidable~\cite{brotherston2010,wendling2010}. Given the wide
applicability of abstract separation logic, even a semi-decision procedure for $\psl$
would assist program verification.

However the definition of separation algebras presented by~\cite{calcagno2007} is not
necessarily canonical.
Most notably~\cite{dockins2009} suggested the
following useful additional properties for spatial reasoning%
  \footnote{Dockins et al.~\cite{dockins2009} also suggested generalising
  separation algebras to have a \emph{set} of units; it is an easy
  corollary of~\cite[Lemma 3.11]{brotherston2013} that single-unit and
  multiple-unit separation algebras satisfy the same set of
  formulae, and we do not pursue this generalisation in this paper.}%
:
\begin{description}
\item[\underline{Indivisible unit}] If combining heap $h_1$ with
  heap $h_2$ gives the empty heap, then $h_1$ and $h_2$ must themselves
  be the empty heap;
\item[\underline{Disjointness}] If the result of combining heap $h_1$ with itself
  is defined, then $h_1$ must be the empty heap;
\item[\underline{Splittability}] Every non-empty heap $h_0$ can be split into two
  non-empty heaps $h_1$ and $h_2$;
\item[\underline{Cross-split}] If a heap can be split in two different ways, then
  there should be heaps that constitute the intersections of these
  splittings.
\end{description}
Conversely, following~\cite{Gotsman:Precision} some authors have further
generalised separation algebras by dropping cancellativity; we will present concrete
examples from~\cite{JensenBirkedal2012,Vafeiadis:Relaxed}. These extensions and
restrictions of $\psl$ point to a need to present \emph{modular} proof theory and proof
search techniques which allow the axiomatic properties of the abstract models to
be adjusted according to the needs of a particular application.

\

This paper is an extended journal version of the conference paper~\cite{hou2013b}.
In that paper, we solved the open problems of presenting sound and complete structural
proof theory, and gave a semi-decision procedure, along with an efficient
implementation, for Propositional Abstract Separation Logic.
We further
showed that our methods could encompass the axiomatic extensions
of~\cite{dockins2009}, and conversely that cancellativity and partial-determinism could
be dropped, and so our proof theory was \emph{modular} in the sense that it could
be used for many neighbouring logics of $\psl$, including BBI. In this journal paper we
make a major extension to the modularity of our approach by introducing a technique to
synthesise proof rules from any spatial axiom in a certain format, general enough to
encompass all axioms of~\cite{dockins2009}. The remainder of this introduction
sketches the techniques used in this paper.

Because of the similarity between non-deterministic monoids, which provide semantics
for BBI~\cite{GalmicheLarcheyWendling2006}, and the separation algebras which
provide semantics for $\psl$, it is natural to investigate whether techniques used
successfully for BBI can be extended to $\psl$. This paper answers this question in
the affirmative by extending the work on BBI of~\cite{hou2013,Hou:Labelled15}. In
these papers, a sound and complete proof theory was provided for BBI in the style of
\emph{labelled sequent calculus}~\cite{negri2001}, a proof style for modal and
substructural logics with Kripke-style frame semantics in which statements about
the elements of the frame are explicitly included in the context of sequents. This allows
relational properties of the semantics to be explicitly represented as proof rules, which
allows labelled sequent calculi to encompass a wide variety of logics in a modular style
-- the addition or subtraction of semantic properties corresponds exactly to the addition
or subtraction of the corresponding proof rules.

This paper builds on~\cite{Hou:Labelled15} by
presenting a labelled sequent calculus for a sublogic $\ls$ of BBI, which is of
no intrinsic interest that we are aware of, but which does include all BI connectives. We
then show that it can be extended to a labelled sequent calculus for BBI, for $\psl$, and
for various neighbouring logics, by extending it with instances of a \emph{general
structural rule} synthesised from axioms on the semantics. This is possible so
long as the axiom is in a certain format, which is sufficiently general to encompass, for
example, the spatial properties identified by~\cite{dockins2009}. We call an axiom in
this format a \emph{frame axiom}.
We then show that our sequent
calculi can be used for effective backward proof search, thereby
providing semi-decision procedures for a variety of logics.  Our
implementation, Separata%
\footnote{Available at \url{http://users.cecs.anu.edu.au/~zhehou}.},
is the first automated theorem prover for $\psl$ and many of its
neighbours.  Separata differs from our previous implementation
$\fvlsbbi$ for BBI in two aspects: first, Separata can handle multiple
abstract separation logics, including BBI, whereas $\fvlsbbi$ is
designed for BBI only; second, Separata is a semi-decision procedure, whereas
$\fvlsbbi$ adopts a heuristic proof search which is incomplete.

In this work, we are interested in proof search procedures that are complete. In this setting,
sequent calculi are amenable to backward proof-search only if the cut
rule is redundant. This result follows much as for the calculus
$LS_{BBI}$ of~\cite{Hou:Labelled15}. However completeness does not follow so
easily; in~\cite{Hou:Labelled15} the
completeness of $LS_{BBI}$ was shown by mimicking derivations in the Hilbert
axiomatisation of BBI. This avenue is no longer viable for $\psl$ because
partial-determinism and cancellativity are not axiomatisable in
BBI~\cite{brotherston2013}. That is, there can be no Hilbert calculus 
in the language of BBI which is sound and
complete with respect to separation algebras.
We instead prove the cut-free completeness of our labelled sequent
calculi via a \emph{counter-model construction} procedure which shows
that if a formula is not cut-free derivable in our sequent calculus
then it is falsifiable in some PASL-model.

The calculi of this paper differ in style from $LS_{BBI}$ because, in~\cite{Hou:Labelled15},
explicit \emph{substitutions} are used in the proof rules, whereas in this paper these
are replaced by explicit \emph{equality} assertions. These are easier for us to manage
with respect to proving the modular completeness of our family of calculi, but the
presentation with substitutions is more amenable to implementation. We hence show
how equivalent new calculi can be defined, with substitutions replacing equalities,
and show how this allows a semi-decision procedure to be implemented.
Experimental results show that our prover is usually faster
than other provers for BBI when tested against the same benchmarks of
BBI formulae.

This paper improves upon all aspects of the presentation of results from its conference
predecessor~\cite{hou2013b}, partly because of lesser limitations on space, but we
here briefly summarise the more important differences between this paper and the
earlier work:
\begin{itemize}
\item A new modular framework of calculi based on frame axioms and
  synthesised structural rules. The completeness of calculi in this
  framework can be obtained in one proof. In the previous work, each
  new calculus required a new proof;
\item A new completeness proof by counter-model construction for a
  framework of calculi. This proof includes treatments for
  splittability and cross-split, which are not included in the
  previous work;
\item A translation from the current calculi to previous calculi with
  global label substitutions;
\item More comprehensive experiments with testing of randomly
  generated formulae;
  \item Many more examples of concrete separation algebras and their applications;
\item Example derivations of various formulae; and a
\item Discussion of applications of this work. 
\end{itemize}

The remainder of this paper is structured as
follows. Section~\ref{sec:bbi_pasl} introduces Propositional Abstract
Separation Logic via its separation algebra semantics, gives a number
of concrete examples of these semantics, and defines the labelled
sequent calculus for $\psl$. Fundamental results such as soundness and
cut-elimination are also
proved. Section~\ref{sec:counter_model_constr} proves the completeness
of our calculi framework by counter-model
construction. Section~\ref{sec:extension_psl} shows how our framework
can encompass various neighbouring logics of $\psl$, based on models
with different spatial properties, and discusses how these properties
manifest in examples.  Section~\ref{sec:eq_subs} shows how to translate our calculi into a format that is more
amenable to implementation, and presents some example derivations.
Section~\ref{sec:experiment} presents the implementation and experiments.  Section~\ref{sec:applications} discusses applications and extensions of the calculi in this work. Finally, Section~\ref{sec:rel_work}
discusses related work.


%% file: psl.tex
\section{A labelled sequent calculus for $\psl$}
\label{sec:bbi_pasl}

In this section we define the \emph{separation algebra} semantics of
Calcagno et al.~\cite{calcagno2007} for Propositional Abstract Separation Logic
($\psl$), present concrete examples of these semantics, and give the labelled
sequent calculus $\lspslh$ for this logic.
Soundness and cut-elimination are then demonstrated for $\lspslh$.


\subsection{Propositional abstract separation logic}
\label{subsec:pasl}
The formulae of $\psl$ are defined inductively as follows, where $p$ ranges over
some set $Var$ of propositional variables:
\begin{align*}
A ::= & \ p\mid \top\mid \bot\mid\lnot A\mid A\lor A\mid A\land A\mid A\limp A\mid \top^*\mid A\mand A\mid A\mimp A
\end{align*}

$\psl$-formulae will be interpreted via the following semantics:

\begin{definition}\label{dfn:sepalg}
A \emph{separation algebra}, or partial cancellative commutative monoid, is a
triple $(H,\circ,\epsilon)$ where $H$ is a non-empty set, $\circ$ is a partial
binary function $H\times H\pfun H$ written infix, and $\epsilon\in H$,
satisfying the following conditions, where `$=$' is interpreted as
``either, both sides are undefined, or, both sides are defined and equal'':
\begin{description}
\item[identity:] $\forall h\in H.\, h\circ\epsilon = h$
\item[commutativity:] $\forall h_1,h_2\in H.\, h_1\circ h_2 = h_2\circ h_1$
\item[associativity:] $\forall h_1,h_2,h_3\in H.\, h_1\circ(h_2\circ h_3) = (h_1\circ h_2)\circ h_3$
\item[cancellativity:] $\forall h_1,h_2,h_3,h_4\in H.$ if $h_1\circ h_2 = h_3$ and $h_1\circ h_4 = h_3$ then $h_2 = h_4$
\end{description}
\end{definition}

Note that the \textit{partial-determinism} of the monoid is assumed since 
$\circ$ is a partial function: 
for any $h_1,h_2,h_3,h_4\in H$, if $h_1\circ h_2 = h_3$ and $h_1\circ h_2 = h_4$ then $h_3 = h_4$.

\begin{example}\label{Ex:Reynolds}
The paradigmatic example of a separation algebra is the set of
\emph{heaps}~\cite{reynolds2002}: finite partial functions from an infinite set of \emph{locations} to a set of \emph{values}.
Then $h_1\circ h_2=h_1\cup h_2$ if $h_1,h_2$ have disjoint domains, and is
undefined otherwise. $\epsilon$ is the empty function.
\end{example}

\begin{example}\label{ex:partcomsemi}
A \emph{partial commutative semigroup}~\cite{Bornat:Permission}, also known
as a \emph{permission algebra}%
\footnote{We prefer the former term, as many interesting examples have little to
do with permissions, and the `permissions algebra' terminology is not used consistently in the
literature; compare~\cite{calcagno2007,Vafeiadis:modular}.}~\cite{calcagno2007},
is a set $V$ equipped with an associative commutative partial binary operator
$\star$, written infix. In other words, it is a separation algebra without the
requirement to have a unit, or to be cancellative.

Fixing such a $(V,\star)$, for which we will give some example definitions
shortly, and given an infinite set of locations $Loc$, we define two finite partial
functions
$h_1,h_2$ from $Loc$ to $V$ to be \emph{compatible} iff for all $l$ in the
intersection of their domains, $h_1(l)\star h_2(l)$ is defined.
We then define the
binary operation $\circ$ on partial functions $h_1,h_2$ as undefined if they are
not compatible. Where they are compatible, $(h_1\circ h_2)(l)$ is defined as:
$$
(h_1\circ h_2)(l)=\begin{cases}
  h_1(l)\star h_2(l) & l\in dom(h_1)\cap dom(h_2) \\
  h_1(l) & l\in dom(h_1)\setminus dom(h_2) \\
  h_2(l) & l\in dom(h_2)\setminus dom(h_1) \\
  \mbox{undefined} & l\notin dom(h_1)\cup dom(h_2)
\end{cases}
$$

Setting $\epsilon$ as the empty function, many examples of concrete separation
algebras have this form, with the $\star$ operation, where defined, intuitively
corresponding to some notion of \emph{sharing} of resources. 
The following are some example definitions of such a construction:
\begin{itemize}
  \item
Heaps: let $V$ be the set of values, and $\star$ be undefined everywhere.
  \item
Fractional permissions~\cite{Boyland:Checking}: let $V$ be the set of pairs of
values (denoted by $v,w$) and (real or rational) numbers (denoted by $i,j$) in the interval $(0,1]$, and
$$
(v,i)\star(w,j)=\begin{cases}
  (v,i+j) & v=w\mbox{ and }i+j\leq 1 \\
  \mbox{undefined} & \mbox{otherwise}
\end{cases}
$$
  \item
Named permissions~\cite{parkinson2005}: given a set $\mathbb{P}$ of
\emph{permission names}, let $V$ be the set of pairs of values (denoted by $v,w$) and non-empty subsets (denoted by $P,Q$) of $\mathbb{P}$, and
$$
(v,P)\star(w,Q)=\begin{cases}
  (v,P\cup Q) & v=w\mbox{ and }P\cap Q=\emptyset \\
  \mbox{undefined} & \mbox{otherwise}
\end{cases}
$$
  \item
Counting permissions~\cite{Bornat:Permission}: let $V$ be the set of pairs of
values (denoted by $v,w$) and integers (denoted by $i,j$). Here $0$ is interpreted as \emph{total permission},
negative integers as \emph{read permissions}, and positive integers as counters
of the number of permissions taken. Let
$$
(v,i)\star(w,j)=\begin{cases}
  (v,i+j) & v=w\mbox{ and } i<0\mbox{ and }j<0 \\
  (v,i+j) & v=w\mbox{ and } i+j\geq 0\mbox{ and }(\,i<0\mbox{ or }j<0\,) \\
  \mbox{undefined} & \mbox{otherwise}
\end{cases}
$$
  \item
Binary Tree Share Model~\cite{dockins2009}: Consider the set of finite
non-empty binary trees whose leaves are labelled true ($\top$) or false ($\bot$),
modulo the smallest congruence such that
\[
\xymatrix@=0mm{
  & & & \bullet \ar[ddl] \ar[ddr] & & & & & & \bullet \ar[ddl] \ar[ddr] \\
  \top & \sim & & & & \qquad\qquad \mbox{and} \qquad\qquad & \bot & \sim \\
  & & \top & & \top & & & & \bot & & \bot
}
\]
Let $\lor$ (resp. $\land$) be the pointwise
disjunction (resp. conjunction) of representative trees of the same shape. Then
let $V$ be the pairs of values (denoted by $v,w$) and equivalence
classes of trees (denoted by $t,u$) so defined, and with
$\star$ defined as shown below, 
where $[\bot]$ is the equivalence class containing the tree whose only
node contains $\bot$:
$$
(v,t)\star(w,u)=\begin{cases}
  (v,t\lor u) & v=w\mbox{ and } t\land u=[\bot] \\
  \mbox{undefined} & \mbox{otherwise.}
\end{cases}
$$
\end{itemize}

\end{example}

Note that the construction above with partial commutative semigroups does
not in general guarantee cancellativity of the separation algebra; for this we
need to require further that $(V,\star)$ is cancellative and has no idempotent
elements (satisfying $v\star v=v$). As we will see later, some interesting concrete
models fail this requirement, and so we will generalise the results of the paper
to drop cancellativity in Section~\ref{sec:extension_psl}.

\begin{example}
Other concrete separation algebras resemble the construction
of Example~\ref{ex:partcomsemi} without fitting it precisely:
\begin{itemize}
  \item
Finite set of locations: The concrete memory model of a 32-bit
machine~\cite{Jensen:High} has as its locations the set of integers
$[0\ldots2^{32})$.
  \item
Total functions: Markings of Petri nets~\cite{Murata:Petri} without capacity
constraints are simply multisets. They may be considered as separation algebras~\cite{calcagno2007} by taking $Loc$ to be $Places$ and
$(V,\star)$ to be the set of natural numbers with addition, then considering the
set of \emph{total} functions $Places\to\mathbb{N}$, with $\circ$ defined as
usual (hence, as multiset union), and $\epsilon$ as the constant $0$ function. If
there is a global capacity constraint $\kappa$ then we let $i\star j$ be undefined 
if $i+j>\kappa$, and hence $\circ$ becomes undefined also in the usual way.

Note that this example can only be made to exactly fit the construction of
Example~\ref{ex:partcomsemi} if we restrict ourselves to \emph{markings of
infinite Petri nets with finite numbers of tokens}. In this case we would consider
a place without tokens to have an undefined map, rather than map to $0$, and
set $V$ to be the positive integers.
  \item
Constraints on functions: The endpoint heaps of~\cite{Villard:Proving} are only
those partial functions that are \emph{dual}, \emph{irreflexive} and
\emph{injective} (we refer to the citation for the definition of these
properties). Similarly, if the places of a Petri net comes equipped with a capacity
constraint function $\kappa:Places\to\mathbb{N}$, we consider only those
functions compatible with those constraints.
\end{itemize}
\end{example}

The 
examples above, which we do not claim to be exhaustive,
justify the study of the abstract properties shared by these concrete semantics.
We hence now turn to the logic $\psl$, which has semantics in any separation
algebra (Definition~\ref{dfn:sepalg}). In this paper we prefer to express $\psl$
semantics in the style of \emph{ternary relations}, which are standard in
substructural logic~\cite{Anderson:Entailment} and in harmony with the most
important work preceding this paper~\cite{Hou:Labelled15}.
We give the ternary relations version of Definition~\ref{dfn:sepalg}, easily seen to
be equivalent, as follows.

\begin{definition}
\label{def:kripke_rel_frame}
A \emph{$\psl$ Kripke relational frame} is a triple $(H, R, \epsilon)$, where $H$
is a non-empty set of \emph{worlds}, $R \subseteq H \times H \times H$, and
$\epsilon \in H$, satisfying the following conditions for all $h_1,h_2,h_3,h_4,
h_5$ in $H$:
\begin{description}
\item[identity:] $R(h_1,\epsilon, h_2) \Leftrightarrow$ $h_1 = h_2$
\item[commutativity:] $R(h_1,h_2, h_3) \Leftrightarrow$ $R(h_2,h_1,
  h_3)$
\item[associativity:] $(R(h_1,h_5, h_4) ~~\&~~ R(h_2,h_3,h_5)) \Rightarrow$
 $\exists h_6.(R(h_6,h_3, h_4) ~~\&~~ R(h_1,h_2, h_6))$
\item[cancellativity:] $(R(h_1,h_2,h_3) ~~\&~~ R(h_1,h_4,h_3)) \Rightarrow h_2 = h_4$
\item[partial-determinism:] $(R(h_1,h_2,h_3) ~~\&~~ R(h_1,h_2, h_4)) \Rightarrow h_3 = h_4$.
\end{description}
\end{definition}
%
A \emph{$\psl$ Kripke relational model} is a tuple
$\mymodel{H}{R}{\epsilon}{\val}$ of a $\psl$ Kripke relational frame $(H,R,
\epsilon)$ and a {\em valuation} function $\val : Var \rightarrow \Pcal(H)$
(where $\Pcal(H)$ is the power set of $H$).  The forcing relation $\Vdash$
between a model  $\Mcal = \mymodel{H}{R}{\epsilon}{\val}$
and a formula is defined in Table~\ref{tab:psl_semantics}, where we write
$\Mcal, h \not \Vdash A$ for the negation of $\Mcal, h \Vdash A$. Given a model
$\Mcal = \mymodel{H}{R}{\epsilon}{\val}$, a formula is \emph{true at (world)
$h$} iff $\Mcal,h \Vdash A$. The formula $A$ is {\em valid} iff it is true at
all worlds of all models.
\begin{table*}
{
\centering
\begin{tabular}{l@{\extracolsep{1cm}}l}
\begin{tabular}{lcl}
$\Mcal, h \Vdash p$ & iff & $p \in Var$ and $h \in v(p)$
\\
$\Mcal, h \Vdash A\land B$ & iff & $\Mcal, h \Vdash A$ and $\Mcal, h \Vdash B$
\\
$\Mcal, h \Vdash A\limp B$ & iff & $\Mcal, h\not \Vdash A$ or $\Mcal, h\Vdash B$\\
$\Mcal, h \Vdash A\lor B$ & iff & $\Mcal, h \Vdash A$ or $\Mcal, h \Vdash B$
\\
\end{tabular}
&
\begin{tabular}{lcl}
$\Mcal, h \Vdash \top^*$ & iff & $h = \epsilon$\\
$\Mcal, h \Vdash \top$ & iff & always
\\
$\Mcal, h \Vdash \bot$ & iff & never 
\\
$\Mcal, h \Vdash \lnot A$ & iff & $\Mcal, h \not \Vdash A$
\end{tabular}\\
\multicolumn{2}{l}{
\hspace{-8px}
\begin{tabular}{lcl}
$\;\;\Mcal, h \Vdash A\mand B$ & iff & 
$\exists h_1,h_2.(R(h_1,h_2, h)$ and $\Mcal, h_1\Vdash A$ and $\Mcal, h_2 \Vdash B)$\\
$\;\;\Mcal, h \Vdash A \mimp B$ & iff & 
$\forall h_1,h_2.((R(h,h_1,h_2)$ and $\Mcal, h_1 \Vdash A)$ implies $\Mcal, h_2 \Vdash B)$
\end{tabular}
}
\end{tabular}
}
\caption{Semantics of $\psl$, where $\Mcal = \mymodel{H}{R}{\epsilon}{\val}.$}
\label{tab:psl_semantics}
\end{table*}


\subsection{The labelled sequent calculus $\lspslh$}
\label{subsec:lspsl}

Let $\LVar$ be an infinite set of {\em label variables}, and let the
set $\Lcal$ of \emph{labels} be $\LVar\cup\{\epsilon\}$, where
$\epsilon$ is a label constant not in $\LVar$; here we overload the
notation for the identity world in the semantics. Labels will be
denoted by lower-case letters such as $a,b,x,y,z$.  A {\em labelled
  formula} is a pair $a : A$ of a label $a$ and formula $A$.  As usual
in a labelled sequent calculus, one needs to incorporate Kripke
relations explicitly into the sequents. This is achieved via the
syntactic notion of {\em relational atoms}, which have the form of
either $a = b$ (\emph{equality}), $a \not = b$ (\emph{inequality}), or
a \emph{ternary relational atom} $(a, b \simp c)$ standing for
$R(a,b,c)$, where $a,b,c$ are labels.  A \emph{sequent} takes the form
$$\myseq{\Gcal}{\Gamma}{\Delta}$$
where $\Gcal$ is a set of
relational atoms, and $\Gamma$ and $\Delta$ are \emph{sets} of labelled
formulae. We also use the symbol ``$;$'' inside $\Gcal$, $\Gamma$ and $\Delta$
to indicate set union: for example,
$\Gamma;A$ is $\Gamma\cup\{A\}$.  Given
$\Gcal$, we denote by $E(\Gcal)$ the set of equations occurring in
$\Gcal$.

We now abuse the sequent turnstile slightly to write 
$E \vdash s = t$, where $E$ is a (possibly
infinite) set of equations, to denote an \emph{equality judgment} under
the assumption $E$, defined inductively as follows:
$$
\mbox{
\AxiomC{$(s = t) \in E$}
\UnaryInfC{$E \vdash s = t$}
\DisplayProof
}
\qquad
\mbox{
\AxiomC{}
\UnaryInfC{$E \vdash s = s$}
\DisplayProof
}
\qquad
\mbox{
\AxiomC{$E \vdash s = t$}
\UnaryInfC{$E \vdash t = s$}
\DisplayProof
}
\qquad
\mbox{
\AxiomC{$E \vdash s = t$}
\AxiomC{$E \vdash t = u$}
\BinaryInfC{$E \vdash s = u$}
\DisplayProof
}
$$
It is easy to see that $E \vdash s = t$ iff $E' \vdash s = t$ for a finite subset $E'$ of $E.$
Note that an equality judgement is not a sequent but abuses the
sequent turnstile to keep track of equalities.

As we shall soon see, working within a labelled sequent calculus
framework allows us to synthesise, in a generic way, proof rules that
correspond to a variety of different properties of separation
algebras, and their extensions. However we first must introduce a
core logic, a sublogic of BBI (and hence, of $\psl$) which we call
$\ls$, which consists only of identity, cut, logical rules and the
structural rules $NEq$ and $EM$. The proof system for this sublogic
is presented in Figure~\ref{fig:LS}. The structural rule $EM$ is
essentially a form of cut on equality predicates. The rules $NEq$ and
$EM$ are admissible for $\ls$ and many of its extensions, but will be
needed for some extensions, such as the extension of $\psl$ with
\emph{splittability}. Note that the equality judgment $E(\Gcal)\vdash
w = w'$ is not a premise requiring proof, but rather a condition for
the rule id. Therefore the rules $id$, $\bot L$, $\top R$, $\top^*
R$, $NEq$ are \emph{zero-premise} rules. In the rules $\mand R$ and
$\mimp L$, the respective principal formulae $z: A \mand B$ and $y: A
\mimp B$ also occur in the premises. This is to ensure that contraction is
admissible, which is essential to obtain cut-elimination.

\begin{figure*}
\footnotesize
\centering
\begin{tabular}{cc}
\multicolumn{2}{l}{\textbf{Identity and Cut:}}\\[10px]
\AxiomC{$E(\Gcal) \vdash w = w'$}
\RightLabel{\tiny $id$}
\UnaryInfC{$\myseq{\Gcal}{\Gamma;w:p}{w':p;\Delta}$}
\DisplayProof
&
\AxiomC{$\myseq{\Gcal}{\Gamma}{x:A;\Delta}$}
\AxiomC{$\myseq{\Gcal}{\Gamma;x:A}{\Delta}$}
\RightLabel{\tiny $cut$}
\BinaryInfC{$\myseq{\Gcal}{\Gamma}{\Delta}$}
\DisplayProof
\\[15px]
\multicolumn{2}{l}{\textbf{Logical Rules:}}\\[10px]
\AxiomC{$$}
\RightLabel{\tiny $\bot L$}
\UnaryInfC{$\myseq{\Gcal}{\Gamma; w:\bot}{\Delta}$}
\DisplayProof
$\qquad$
\AxiomC{$\myseq{\Gcal;w=\epsilon}{\Gamma}{\Delta}$}
\RightLabel{\tiny $\top^* L$}
\UnaryInfC{$\myseq{\Gcal}{\Gamma;w:\top^*}{\Delta}$}
\DisplayProof
&
\AxiomC{$$}
\RightLabel{\tiny $\top R$}
\UnaryInfC{$\myseq{\Gcal}{\Gamma}{w:\top;\Delta}$}
\DisplayProof
$\qquad$
\AxiomC{$E(\Gcal)\vdash w=\epsilon$}
\RightLabel{\tiny $\top^* R$}
\UnaryInfC{$\myseq{\Gcal}{\Gamma}{w:\top^*;\Delta}$}
\DisplayProof\\[15px]
\AxiomC{$\myseq{\Gcal}{\Gamma;w:A;w:B}{\Delta}$}
\RightLabel{\tiny $\land L$}
\UnaryInfC{$\myseq{\Gcal}{\Gamma;w:A\land B}{\Delta}$}
\DisplayProof
&
\AxiomC{$\myseq{\Gcal}{\Gamma}{w:A;\Delta}$}
\AxiomC{$\myseq{\Gcal}{\Gamma}{w:B;\Delta}$}
\RightLabel{\tiny $\land R$}
\BinaryInfC{$\myseq{\Gcal}{\Gamma}{w:A\land B;\Delta}$}
\DisplayProof\\[15px]
\AxiomC{$\myseq{\Gcal}{\Gamma}{w:A;\Delta}$}
\AxiomC{$\myseq{\Gcal}{\Gamma;w:B}{\Delta}$}
\RightLabel{\tiny $\limp L$}
\BinaryInfC{$\myseq{\Gcal}{\Gamma;w:A\limp B}{\Delta}$}
\DisplayProof
&
\AxiomC{$\myseq{\Gcal}{\Gamma;w:A}{w:B;\Delta}$}
\RightLabel{\tiny $\limp R$}
\UnaryInfC{$\myseq{\Gcal}{\Gamma}{w:A\limp B; \Delta}$}
\DisplayProof\\[15px]
\AxiomC{$\myseq{\Gcal;(x,y \simp z)}{\Gamma;x:A;y:B}{\Delta}$}
\RightLabel{\tiny $\mand L$}
\UnaryInfC{$\myseq{\Gcal}{\Gamma;z:A\mand B}{\Delta}$}
\DisplayProof
&
\AxiomC{$\myseq{\Gcal;(x,z \simp y)}{\Gamma;x:A}{y:B;\Delta}$}
\RightLabel{\tiny $\mimp R$}
\UnaryInfC{$\myseq{\Gcal}{\Gamma}{z:A\mimp B;\Delta}$}
\DisplayProof\\[15px]
\multicolumn{2}{c}{
\AxiomC{$\myseq{\Gcal;(x,y \simp z')}{\Gamma}{x:A;z:A\mand B;\Delta}$}
\AxiomC{$\myseq{\Gcal;(x,y \simp z')}{\Gamma}{y:B;z:A\mand B;\Delta} \qquad E(\Gcal) \vdash z = z'$}
\RightLabel{\tiny $\mand R$}
\BinaryInfC{$\myseq{\Gcal;(x,y \simp z')}{\Gamma}{z:A\mand B;\Delta}$}
\DisplayProof
}\\[15px]
\multicolumn{2}{c}{
\AxiomC{$\myseq{\Gcal;(x,y' \simp z)}{\Gamma;y:A\mimp B}{x:A;\Delta}$}
\AxiomC{$\myseq{\Gcal;(x,y' \simp z)}{\Gamma;y:A\mimp B; z:B}{\Delta} \qquad E(\Gcal) \vdash y = y'$}
\RightLabel{\tiny $\mimp L$}
\BinaryInfC{$\myseq{\Gcal;(x,y' \simp z)}{\Gamma;y:A\mimp B}{\Delta}$}
\DisplayProof
}
\\[15px]
\multicolumn{2}{l}{\textbf{Structural Rules:}}\\[15px]
\AxiomC{$E(\Gcal) \vdash u = v$}
\RightLabel{\tiny $NEq$}
\UnaryInfC{$\myseq {\Gcal; u \not = v} \Gamma \Delta$}
\DisplayProof 
&
\AxiomC{$\myseq{\Gcal; x = y} \Gamma \Delta$}
\AxiomC{$\myseq{\Gcal; x \not = y} \Gamma \Delta$}
\RightLabel{\tiny $EM$}
\BinaryInfC{$\myseq \Gcal \Gamma \Delta$}
\DisplayProof
\\[15px]
\multicolumn{2}{l}{\textbf{Side conditions:}} \\
\multicolumn{2}{l}{In $\mand L$ and $\mimp R$, the labels $x$ and $y$
do not occur in the conclusion. }\\
\end{tabular}
\caption{Inference rules for the core sublogic $\ls$.}
\label{fig:LS}
\end{figure*}

Given a relational frame
$(H, R, \epsilon)$,
a function $\rho: \Lcal \rightarrow H$ from
labels to worlds is a {\em label mapping} 
iff it satisfies $\rho(\epsilon) =  \epsilon$, mapping the
label constant $\epsilon$ to the identity world $\epsilon \in H$. Intuitively, a labelled formula $a:A$ means that formula $A$ is true in world $\rho(a)$.
Thus we define an \textit{extended $\psl$ Kripke relational model} $(H,R,\epsilon,\val,\rho)$ as a model equipped with a label mapping.

\begin{definition}[Sequent Falsifiability]
A sequent $\myseq{\Gcal}{\Gamma}{\Delta}$ is \emph{falsifiable in an extended model}
$\Mcal = \mylmodel{H}{R}{\epsilon}{\val}{\rho}$ if for every $x : A \in \Gamma$, $(a,b\simp c) \in \Gcal$, and for every
$y : B \in \Delta$,
we have each of $\Mcal,\rho(x) \Vdash A$ and 
$R(\rho(a),\rho(b),\rho(c))$ and $\Mcal,\rho(y) \not \Vdash B.$
It is {\em falsifiable} if it is falsifiable in some extended model. 
\end{definition}

\paragraph{Synthesising structural rules from frame axioms}
We now define extensions of the sublogic $\ls$ via first-order axioms that
correspond to various semantic conditions used to define $\psl$ and its
variations. For this work, we consider only axioms that are closed formulae of
the following general axiom form where $k,l,m,n,p$ are natural numbers:
\begin{equation}
\label{eq:ax}
\forall x_1,\dots,x_m. (s_1 = t_1 \,\& \cdots \&\, s_p = t_p 
   \,\&\, S_1 \,\& \cdots \&\, S_k \,\Rightarrow\, \exists y_1, \dots, y_n. (T_1 \,\& \cdots \&\, T_l))
\end{equation}
Note that where $k$, $l$, or $p$ are $0$, we assume the empty conjunction is
$\top$. 
We further require the following conditions:
\begin{itemize}
\item each $S_i$, for $1\leq i \leq k$, is either a ternary relational
  atom or an inequality;
\item each $T_i$, for $1\leq i \leq l$, is a relational atom;
\item every label variable in $\bigcup_{1\leq i \leq k} S_i$ occurs only
  once;
\item if $S_i$, for $1\leq i\leq k$, is a ternary relational atom, then
  $\epsilon$ does not occur in $S_i$.
\end{itemize}

We call axioms of this form \emph{frame axioms}.  A frame axiom can be
given semantics in terms of Kripke frames, following the standard
classical first-order interpretation (see e.g.,~\cite{Fitting1996}).
Recall that a first-order model is a pair $(D, I)$ of a non-empty {\em
  domain} $D$ and an {\em interpretation} function $I$ that associates
each constant in the first-order language to a member of $D$ and every
$n$-ary relation symbol to an $n$-ary relation over $D$.  When
interpreting first-order formulae with free variables, we additionally
need to specify the valuation of the free variables, i.e., a mapping
of the free variables to elements of $D$.  The notion of truth of a
first-order formula (under a model and a given valuation of its free
variables) is standard and the reader is referred to,
e.g.,~\cite{Fitting1996} for details.

\begin{definition}
\label{def:frame-sat}
A Kripke frame $(H, R, \epsilon)$ {\em satisfies} a frame axiom $F$
iff $F$ is true in the first order model $(H, I)$, where $I$ is the
interpretation function that associates the symbol $\epsilon$ to the
label constant $\epsilon$ in the set $H$, the predicate symbol $\simp$ to the relation
$R$, and the equality symbol $=$ to the identity relation over $H$.  A
Kripke frame satisfies a set $\Fcal$ of frame axioms iff it satisfies
every frame axiom in the set.
\end{definition}
%

A frame axiom such as the one from Formula~(\ref{eq:ax}) 
induces the following \emph{general structural rule}:
\[
\mbox{
\AxiomC{$\myseq {\Gcal; S_1; \dots; S_k; T_1; \dots; T_l} \Gamma \Delta$}
\AxiomC{$E(\Gcal) \vdash s_1 = t_1 \qquad \cdots \qquad E(\Gcal) \vdash s_p = t_p$}
\BinaryInfC{$\myseq {\Gcal; S_1; \dots; S_k} \Gamma \Delta$}
\DisplayProof
}
\]
where the existential condition in Equation~\ref{eq:ax} becomes a
side-condition that the existentially quantified variables
$y_1,\dots,y_n$ must be fresh label variables not occurring in the
conclusion of the rule.

\begin{example}
\label{ex:frame-axiom-1}
  The semantic clauses in Definition~\ref{def:kripke_rel_frame} 
  can be captured by the following frame axioms:
  \begin{description}
  \item[identity 1:] $\forall h_1,h_2,h_3. (h_2 = \epsilon$ $\&$ $R(h_1,h_2,h_3)) \Rightarrow$ $h_1 = h_3$
  \item[identity 2:] $\forall h_1,h_2. h_1 = h_2 \Rightarrow$ $R(h_1,\epsilon,h_2)$
  \item[commutativity:] $\forall h_1,h_2,h_3. R(h_1,h_2,h_3) \Rightarrow$ $R(h_2,h_1,h_3)$
  \item[associativity:] $\forall h_1,h_2,h_3,h_4,h_5,h_5'.$ 
  \\ $(h_5 = h_5'$ $\&$ $R(h_1,h_5,h_4)$ $\&$ $R(h_2,h_3,h_5'))
    \Rightarrow$ $\exists h_6.$ $(R(h_6,h_3, h_4)$ $\&$ $R(h_1,h_2, h_6))$
  \item[cancellativity:] $\forall h_1,h_2,h_3,h_1',h_3'.$
    \\ $(h_1 = h_1'
    ~\&~ h_3 = h_3' ~\&~ R(h_1,h_2,h_3) ~\&~ R(h_1',h_4,h_3'))$ 
   $\Rightarrow h_2 = h_4$
  \item[partial-determinism:] $\forall h_1,h_1',h_2,h_2',h_3,h_4.$
   \\ $(h_1 = h_1'$ $\&$ $h_2 = h_2'$ $\&$ $R(h_1,h_2,h_3)$ $\&$ $R(h_1',h_2',h_4))
       \Rightarrow h_3 = h_4$.
  \end{description}
  These frame axioms are mostly a straightforward translation from the semantic clauses of Definition~\ref{def:kripke_rel_frame}
  into the syntactic form, replacing the relation $R$ with the predicate symbol $\simp.$ It is trivial to show that
  the Kripke frames defined in Definition~\ref{def:kripke_rel_frame} satisfy the frame axioms above. 
  However, notice that the syntactic form of the frame axioms does not allow more than one occurrence of a variable
  in the left hand side of the implications. Thus, for each semantic clause of
  Definition~\ref{def:kripke_rel_frame}, we need to identify each 
  world that occurs multiple (say, $n$) times on the left hand side of
  the implications, make $n$ distinct copies of that world, and add equalities
  relating them. If $\epsilon$ occurs in a ternary relational atom
  on the left hand side, we need to create a fresh (universally
  quantified) variable, e.g., $w$, and add that $w = \epsilon$.

  Take the associativity axiom in
  Definition~\ref{def:kripke_rel_frame} as an example:
  \begin{center}
    $\forall h_1,h_2,h_3,h_4,h_5.(R(h_1,h_5, h_4) \,\&\, R(h_2,h_3, h_5)
    \,\Rightarrow\, \exists h_6.(R(h_6,h_3, h_4) \,\&\, R(h_1,h_2, h_6)))$
  \end{center}
  The world $h_5$ occurs twice on the left hand
  side, so we make two copies of it: $h_5$ and $h_5'$. The
  corresponding axiom in frame axiom form is then:
  \begin{center}
    $\forall h_1,h_2,h_3,h_4,h_5.(h_5 = h_5' \,\&\, (h_1,h_5 \simp h_4) \,\&\,
    (h_2,h_3 \simp h_5') \,\Rightarrow\, \exists h_6.((h_6,h_3 \simp h_4) \,\&\,
    (h_1,h_2 \simp h_6)))$
  \end{center}
  We may then synthesise the following structural rule for associativity:
  \begin{center}
    \AxiomC{$\myseq {\Gcal;(u,y \simp x);(v,w \simp y');(z, w \simp x); (u, v \simp
        z)} \Gamma \Delta$ \qquad $E(\Gcal)\vdash y = y'$}

    \RightLabel{$A$}

    \UnaryInfC{$\myseq {\Gcal;(u,y \simp x);(v,w \simp y')} \Gamma \Delta$}

    \DisplayProof
  \end{center}
  with the ``freshness'' side-condition that $z$ does not appear in the conclusion.

  Note that this rule is applicable on $(u,y \simp x),(v,w
  \simp y)$, as $E(\Gcal)\vdash y = y$ is trivial.
\end{example}

From the frame axioms in Example~\ref{ex:frame-axiom-1} above, we obtain the structural rules of
Figure~\ref{fig:sepalg}. The identity axiom, as it is a
bi-implication, gives rise to two rules $E$ and $U$. The commutativity
axiom translates to rule $Com$, associativity to $A$, cancellativity to
$C$ and partial determinism to $P$. The proof system $\lspslh$ is
defined to be the rules of Figure~\ref{fig:LS} for the sublogic $\ls$, plus the
synthesised structural rules of Figure~\ref{fig:sepalg}.

\begin{figure}
\begin{tabular}{cc}
  \AxiomC{$\myseq {\Gcal;(x, y \simp z);x = z } \Gamma \Delta$}
  \AxiomC{$E(\Gcal)\vdash y = \epsilon$}
  \RightLabel{\tiny $E$}
  \BinaryInfC{$\myseq {\Gcal;(x, y \simp z)} \Gamma \Delta$}
  \DisplayProof
  &
  \AxiomC{$\myseq {\Gcal; (x, \epsilon \simp y)} \Gamma \Delta$}
  \AxiomC{$E(\Gcal) \vdash x = y$}
  \RightLabel{\tiny $U$}
  \BinaryInfC{$\myseq \Gcal \Gamma \Delta$}
  \DisplayProof
\\[15px]
\multicolumn{2}{c}{
  \AxiomC{$\myseq {\Gcal; (x,y \simp z);(y,x \simp z)} \Gamma \Delta$}
  \RightLabel{\tiny $Com$}
  \UnaryInfC{$\myseq {\Gcal; (x,y \simp z)} \Gamma \Delta$}
  \DisplayProof
}
\\[15px]
\multicolumn{2}{c}{
    \AxiomC{$\myseq {\Gcal;(u,y \simp x);(v,w \simp y');(z, w \simp x); (u, v \simp
        z)} \Gamma \Delta$ \qquad $E(\Gcal)\vdash y = y'$}
    \RightLabel{\tiny $A$}
    \UnaryInfC{$\myseq {\Gcal;(u,y \simp x);(v,w \simp y')} \Gamma \Delta$}
\DisplayProof
}\\[15px]
\multicolumn{2}{c}{
\AxiomC{$\myseq {\Gcal;(x,y \simp z); (x', w \simp z'); y = w} \Gamma \Delta$ $\qquad$ $E(\Gcal)\vdash x = x'$ $\qquad$ $E(\Gcal)\vdash z = z'$}
\RightLabel{\tiny $C$}
\UnaryInfC{$\myseq {\Gcal;(x,y \simp z); (x', w \simp z')} \Gamma \Delta$}
\DisplayProof
}\\[15px]
\multicolumn{2}{c}{
\AxiomC{$\myseq {\Gcal;(w,x \simp y); (w',x'\simp z); y = z} \Gamma \Delta$ $\qquad$ $E(\Gcal)\vdash w = w'$ $\qquad$ $E(\Gcal)\vdash x = x'$} 
\RightLabel{\tiny $P$}
\UnaryInfC{$\myseq {\Gcal;(w,x \simp y); (w',x'\simp z)} \Gamma \Delta$}
\DisplayProof
}\\[15px]
\multicolumn{2}{l}{\textbf{Side conditions:}} \\
\multicolumn{2}{l}{In $A$, the label $z$ does not occur in the conclusion. }\\
\end{tabular}
\caption{Structural rules synthesised from frame axioms. Together with
Figure~\ref{fig:LS}, these rules define the sequent calculus $\lspslh$ for the
logic $\psl$.}
\label{fig:sepalg}
\end{figure}

We remark here that it is not always obvious what the effect of each semantic
property will be on the set of valid formulae; for
example it was only recently discovered~\cite{larcheywendling2014} that
cancellativity does not affect validity in the presence of the other properties.
This lends
weight to the suggestion of~\cite{Gotsman:Precision} that cancellativity should
be omitted from the definition of separation algebra; see also our examples of
concrete separation algebras without cancellativity in
Section~\ref{sec:nocancel}. It is nonetheless harmless to include it in our
rules, and may be useful for some extensions of $\psl$, as we discuss in
Section~\ref{sec:rel_work}.

It is easy to check that the following hold:

\begin{theorem}[Soundness of the general structural rule]
  \label{thm:soundness}
Every synthesised instance of the general structural rule is sound with respect to
the Kripke relational frames with the corresponding frame axiom. 
\end{theorem}

\begin{corollary}[Soundness of $\lspslh$]
For any formula $A$, and
for an arbitrary label $w$,
if the labelled sequent 
$\vdash w:A$ is derivable in $\lspslh$
then
$A$ is valid.
\end{corollary}

Note that the soundness of $\lspslh$ has been formally verified via the interactive theorem prover
Isabelle~\cite{Separata-AFP}.

We will give the name $\lsg$ to the general proof system that extends
the rules of Figure~\ref{fig:LS} with any set of structural rules
synthesised from frame axioms. A proof system consisting of the rules
in Figure~\ref{fig:LS} plus a finite number of instances of the
general structural rule is called an \emph{instance} of $\lsg$.


\subsection{Cut-elimination for the general proof system $\lsg$}

In this section we see that the $cut$ rule of Figure~\ref{fig:LS} is
admissible in the general nested sequent calculus $\lsg$. Since
cut-admissibility can be obtained indirectly from the cut-free
completeness proof in the next section, we do not give full details
here.  

A {\em label substitution} is a mapping from label variables to
labels. The {\em domain} of a substitution $\theta$ is the set $\{x
\mid \theta(x) \not = x\}$. We restrict to substitutions with only
finite domains.  We use the notation $[a_1/x_1, \ldots, a_n / x_n]$ to
denote a substitution mapping variables $x_i$ to labels $a_i.$
Application of a substitution $\theta$ to a term or a formula is
written in a postfix notation, e.g., $F[a/x]$ denotes a formula
obtained by substituting $a$ for every free occurrence of $x$ in $F.$
This notation generalises straightforwardly to applications of
substitutions to (multi)sets of formulas, relational atoms and
sequents.

We will first present a substitution lemma for instance systems
of $\lsg$. This requires the following lemma.

\begin{lemma}[Substitution in equality judgments]
Given any set $E$ of equality relational atoms, any labels $x$, $y$ and $z$, and any label variable $w$, if $E\vdash x = y$, then $E[z/w] \vdash (x = y)[z/w]$, where every occurrence of $w$ is replaced with $z$. 
\end{lemma}

In the substitution lemma below we use $ht(\Pi)$ to denote the height of the
derivation $\Pi$.

\begin{lemma}[Substitution for $\lsg$]
In any instance system of $\lsg$, if $\Pi$ is a derivation for the sequent $\myseq{\Gcal}{\Gamma}{\Delta}$, then there is a derivation $\Pi'$ of the sequent $\myseq{\Gcal[y/x]}{\Gamma[y/x]}{\Delta[y/x]}$ where every occurrence of label variable $x$ is replaced by label $y$, such that $ht(\Pi')\leq ht(\Pi)$.
\end{lemma}

Since $\lsg$ does not involve explicit label substitutions in the rules anymore, the proof for the substitution lemma is actually simpler than the proof for $\lsbbi$~\cite{Hou:Labelled15}, to which we refer interested readers.

The admissibility of weakening for any instance system of $\lsg$ can
be proved by a simple induction on the length of the derivation.
The invertibility of the inference rules in $\lsg$ can be proved in a
similar way as for $\lsbbi$~\cite{Hou:Labelled15}, which uses similar techniques
as for $G3c$~\cite{negri2001}. The proofs for the following lemmas 
are a straightforward adaptation of similar
proofs from \cite{Hou:Labelled15} so we omit details here.

\begin{lemma}[Weakening admissibility of $\lsg$]
If $\myseq{\Gcal}{\Gamma}{\Delta}$ is derivable in any instance system of $\lsg$, then for any set $\Gcal'$ of relational atoms, and any set $\Gamma'$ and $\Delta'$ of labelled formulae, the sequent $\myseq{\Gcal;\Gcal'}{\Gamma;\Gamma'}{\Delta;\Delta'}$ is derivable with the same height in that instance of $\lsg$.
\end{lemma}

\begin{lemma}[Invertibility of rules in $\lsg$]
In any instance system of $\lsg$, if $\Pi$ is a cut-free derivation of the conclusion of a rule, then there is a cut-free derivation for each premise, with height at most $ht(\Pi)$.
\end{lemma}

Since the sequents in our definition consists of sets, the
admissibility of contraction is trivial and we do not state it as a
lemma here. The cut-elimination proof here is an adaptation of that
of~\cite{Hou:Labelled15}. The proof in our case is simpler, as our cut rule
does not split context, and our inference rules do not involve
explicit label substitutions. We hence state the theorem here without
proof:

\begin{theorem}[Cut-elimination for $\lsg$]
For any instance of $\lsg$, if a formula is derivable in that
instance, then it is derivable without using $cut$ in that instance.
\end{theorem}


%% file: counter_model_psl.tex
\section{Counter-model construction for $\lsg$}
\label{sec:counter_model_constr}


We now give a counter-model construction procedure that works for all
finite instances (systems with finite rules) of the general proof system $\lsg$,
and hence establishes their completeness.

As the counter-model construction involves infinite sets and sequents,
we extend the definition of equality judgment:

\begin{definition}
Given a (possibly infinite) set $\Gcal$ of relational atoms,
the judgment 
$E(\Gcal)\vdash x = y$ holds iff 
$E(\Gcal_f)\vdash x = y$ holds for some finite $\Gcal_f\subseteq \Gcal$.
\end{definition}

Given a set $\Gcal$ of relational atoms, we define the relation
$=_\Gcal$ as follows: $a =_\Gcal b$ iff $E(\Gcal) \vdash a = b$.  We
next state a lemma which is an immediate result from our equality
judgment rules and will be useful in our counter-model construction
later:

\begin{lemma}
\label{lm:vdashe_eq}
Given a set $\Gcal$ of relational atoms, the relation $=_\Gcal$ 
is an equivalence relation on the set of labels.
\end{lemma}

The equivalence relation $=_\Gcal$ partitions $\Lcal$ into
equivalence classes $[a]_\Gcal$ for each label $a \in \Lcal$:
$$[a]_\Gcal = \{ a' \in \Lcal \mid a =_\Gcal a' \}.$$


The counter-model construction is essentially a procedure to saturate
a sequent by applying all backward applicable rules repeatedly. The
aim is to obtain an infinite saturated sequent from which a
counter-model can be extracted.  We first define a list of required
conditions for such an infinite sequent which would allow the
counter-model construction.

\begin{definition}[General Hintikka sequent]
\label{definition:hintikka_seq}
A labelled sequent $\myseq{\Gcal}{\Gamma}{\Delta}$ is a {\em general Hintikka sequent}
if it satisfies the following
conditions for any formulae $A,B$ and any labels $a,b,c,d,e,z\in\Lcal$:
\begin{enumerate} 

\item It is not the case that $a:A\in \Gamma$, $b:A\in\Delta$ and $a =_\Gcal b.$

\item $a : \bot \not \in \Gamma$ and $a: \top \not \in \Delta.$

\item If $a:\top^*\in \Gamma$ then $a =_\Gcal \epsilon.$

\item If $a:\top^* \in\Delta$ then $a \not =_\Gcal \epsilon.$

\item If $a:A\land B\in \Gamma$ then $a:A\in \Gamma$ and $a:B\in \Gamma.$

\item If $a:A\land B\in \Delta$ then $a:A\in \Delta$ or $a:B\in \Delta.$

\item If $a:A\limp B\in \Gamma$ then $a:A\in \Delta$ or $a:B\in\Gamma.$

\item If $a:A\limp B\in \Delta$ then $a:A\in \Gamma$ and $a:B\in\Delta.$

\item If $z:A\mand B\in \Gamma$ then $\exists x,y,z'$ s.t. 
$(x,y\simp z')\in\Gcal$, $z =_\Gcal z'$, $x:A\in\Gamma$ and $y:B\in\Gamma.$

\item If $z:A\mand B\in \Delta$ then $\forall x,y,z'$ 
if $(x,y\simp z')\in\Gcal$ and $z =_\Gcal z'$ then $x:A\in\Delta$ or $y:B\in\Delta.$

\item If $z:A\mimp B\in \Gamma$ then $\forall x,y,z'$ 
if $(x,z'\simp y)\in\Gcal$ and $z =_\Gcal z'$, then $x:A\in\Delta$ or $y:B\in\Gamma.$

\item If $z:A\mimp B\in \Delta$ then $\exists x,y,z'$ s.t. 
$(x,z'\simp y)\in\Gcal$, $z =_\Gcal z'$, $x:A\in\Gamma$ and $y:B\in\Delta.$

\item It is not the case that $a\neq b\in\Gcal$ and $a=_\Gcal b$.

\item Either $a\neq b \in \Gcal$ or $a=_\Gcal b$.


\item Given a frame axiom of the form
  \[
    \forall x_1,\dots,x_m. (s_1 = t_1 \,\& \cdots \&\, s_p = t_p 
   \,\&\, S_1 \,\& \cdots \&\, S_k \,\Rightarrow\, \exists y_1, \dots, y_n. (T_1 \,\& \cdots \&\, T_l))
  \]
  for any labels $a_1,\cdots,a_m\in\Lcal$ and substitution $\theta = [a_1/x_1,\cdots,a_m/x_m]$, 
  if $\bigcup_{1\leq i \leq k} \{S_i\theta\} \subseteq \Gcal$ and $s_i\theta =_\Gcal t_i\theta$ 
  for $1\leq i \leq p$,
  then there exist $b_1,\cdots, b_n\in\Lcal$ and substitution $\sigma = [b_1/y_1, \dots, b_n/y_n]$
  such that 
  $\bigcup_{1\leq i\leq l} \{(T_i\theta)\sigma\} \subseteq \Gcal$.




\end{enumerate}
\end{definition}

In condition 15, the variables $x_1,\cdots,x_m$ and $y_1,\cdots, y_n$
are schematic variables, i.e., symbols that belong to the
metalanguage, and the substitutions $\theta$ and $\sigma$ replace
these schematic variables with labels. Since $\theta$ and $\sigma$
have disjoint domains, we have that $(x_i\theta)\sigma = x_i\theta$
and $(y_i\theta)\sigma = y_i\sigma$. These will be useful in the
proofs below.

We are often interested in some particular Hintikka sequents that
correspond to certain frame axioms. Given a set $\Frcal$ of frame
axioms, a \emph{$\Frcal$-Hintikka sequent} is an instance of the
general Hintikka sequent where condition 15 holds for each frame axiom
in $\Frcal$. We say a Kripke frame satisfies $\Frcal$ when every frame
axiom in $\Frcal$ is satisfied by the Kripke frame.

Next we show a parametric Hintikka lemma: a Hintikka sequent
parameterised over a set of frame axioms gives a Kripke relational
frame where the set of frame axioms are satisfied and the formulae in
the right hand side of the sequent are false.

\begin{lemma}
\label{lem:hintikka_sat}
Given a set $\Frcal$ of frame axioms, every $\Frcal$-Hintikka
sequent is falsifiable in some Kripke frame satisfying $\Frcal.$ 
\end{lemma}

\begin{proof}
  Let $\myseq{\Gcal}{\Gamma}{\Delta}$ be an arbitrary
  $\Frcal$-Hintikka sequent. We construct an extended model $\Mcal
  = (H,R,\epsilon_\Gcal,v,\rho)$ as follows:
  \begin{itemize}
  \item $H = \{[a]_\Gcal~\mid~a\in\Lcal\};$
  \item $R([a]_\Gcal,[b]_\Gcal,[c]_\Gcal)$ iff $\exists
    a',b',c'$ s.t. $(a',b'\simp c')\in\Gcal, a =_\Gcal a',b =_\Gcal
    b', c=_\Gcal c';$
  \item $\epsilon_\Gcal = [\epsilon]_\Gcal;$
  \item $v(p) = \{[a]_\Gcal~\mid~a:p\in\Gamma\}$ for every $p\in Var$;
  \item $\rho(a) = [a]_\Gcal$ for every $a\in\Lcal$.
  \end{itemize}
  To reduce clutter, we shall drop the subscript $\Gcal$ in
  $[a]_\Gcal$.

  We first show that the $\Frcal$-Hintikka sequent
  $\myseq{\Gcal}{\Gamma}{\Delta}$ gives rise to a Kripke relational
  frame that satisfies the frame axioms in $\Frcal$.
  Take an arbitrary frame axiom $F \in \Frcal$ of form
   \[
    \forall x_1,\dots,x_m. (s_1 = t_1 \,\& \cdots \&\, s_p = t_p 
   \,\&\, S_1 \,\& \cdots \&\, S_k \,\Rightarrow\, \exists y_1, \dots, y_n. (T_1 \,\& \cdots \&\, T_l))
  \]
  We have to show that the frame axiom $F$ above is true in the
  first-order model $(H, I)$, where $I$ is the interpretation function
  such that $\epsilon^I = [\epsilon]_\Gcal$ and $\simp^I = R$.
  That is, for an arbitrary first-order valuation
  $\mu = \{x_1 \mapsto [a_1], \dots, x_m \mapsto [a_m] \}$, we assume
  that $s_i = t_i$ for $1\leq i \leq q$ and $S_i$ holds for
  $1\leq i \leq k$, in the first-order interpretation. We then show
  that $\mu$ can be extended with some valuation for the variables
  $y_1,\cdots,y_n$ such that under the new valuation, 
  $T_i$ holds $1\leq i \leq l$ in the first-order interpretation.

  We construct a series of substitutions as follows:
  \begin{itemize}
  \item
    We start with the substitution $\theta = [a_1/x_1,\cdots,
      a_n/x_n]$. We have that for each $1\leq i \leq p$, $s_i\theta
    =_\Gcal t_i\theta$ from the assumption that $s_i = t_i$ holds
    under $\mu$.

  \item
    For each $1\leq i \leq k$, if $S_i$ has the form $x_u \neq x_v$
    where $u,v\in\{1,\cdots,m\}$, then by assumption, $[a_u] \neq
    [a_v]$ holds. By condition (13) and (14) in
    Definition~\ref{definition:hintikka_seq}, either $x_u\theta
    =_\Gcal x_v\theta$ or $x_u\theta\neq x_v\theta \in \Gcal$. Since
    the former contradicts with the assumption, we have $x_u\theta
    \neq x_v\theta \in\Gcal$.

  \item
    For each $1\leq i \leq k$, if $S_i$ has the form $(x_u,x_v\simp
    x_w)$ where $u,v,w\in\{1,\cdots, m\}$, then by assumption,
    $R([a_u],[a_v],[a_w])$ holds. This means that there are some
    $a_u',a_v',a_w'$ such that $a_u =_\Gcal a_u'$, $a_v =_\Gcal a_v'$,
    $a_w =_\Gcal a_w'$, and $(a_u',a_v'\simp a_w')\in\Gcal$. We
    construct a new substitution $\theta'$ where $\theta'$ is the same
    as the previous substitution $\theta_0$ except $\theta'$ replaces
    $x_u,x_v,x_w$ respectively with $a_u',a_v',a_w'$. The following
    facts hold under the new substitution $\theta'$: (1)
    $S_i\theta'\in\Gcal$. (2) For any $1\leq q\leq p$, $s_q\theta_0
    =_\Gcal t_q\theta_0$ implies $s_q\theta'=_\Gcal t_q\theta'$. If
    neither of $s_q$ and $t_q$ are one of $x_u,x_v,x_w$, then the
    equation is syntactically the same as before. Otherwise, if, for
    example, $s_q$ is $x_u$, then it is replaced by $a_u'$ under
    $\theta'$. But since $a_u' =_\Gcal a_u$, by transitivity of
    $=_\Gcal$ (Lemma~\ref{lm:vdashe_eq}), we obtain the equivalence in
    the first-order interpretation. (3) Since we assume that each
    label variable only occurs once in $\bigcup_{1\leq i \leq k} S_i$,
    and that the constant $\epsilon$ does not occur in ternary
    relations, for any $1\leq q \leq k$ such that $q\neq i$,
    $S_q\theta'$ is syntactically equivalent to $S_q\theta_0$,
    therefore if $S_q\theta_0\in\Gcal$ then $S_q\theta'\in\Gcal$.
  \end{itemize}

  Suppose we start with $\theta$ and iteratively construct a new
  substitution as the last case of the above, and obtain a final
  substitution $\theta''$. It is easy to establish that for each
  $1\leq i \leq p$, $s_i\theta'' =_\Gcal t_i\theta''$, and for each
  $1\leq i \leq k$, $S_i\theta''\in\Gcal$. Therefore by condition (15)
  of Definition~\ref{definition:hintikka_seq}, there exist
  $b_1,\cdots,b_n\in\Lcal$ and a substitution $\sigma =
  [b_1/y_1,\cdots,b_n/y_n]$ such that $\bigcup_{1\leq i\leq l}
  \{(T_i\theta'')\sigma\} \subseteq \Gcal$.  Note also that $x_i\theta
  =_\Gcal x_i\theta''$ for each $1\leq i \leq m$. Now we show that for
  each $1\leq i \leq l$, $T_i$ holds in the first-order model. This is
  done by extending the first-order valuation $\mu$ to $\mu'$, where
  $\mu'$ is the same as $\mu$ except $\mu'$ maps each $y_i$ to $[b_i]$
  for $1\leq i\leq n$.  We consider the following cases depending on
  the shape of $T_i$:

  \begin{itemize}
  \item
    $T_i$ has the form $(w = w')$.
    \begin{itemize}
    \item
      Suppose further that $w$ is $x_u$ and $w'$ is $x_v$ for some
      $u,v\in\{1,\cdots,m\}$. We need to show that $x_u = x_v$ under
      the valuation $\mu'$. Since $\mu'$ only differs from $\mu$ in
      the mappings of $\{y_1,\cdots,y_n\}$, we only need to show that
      $[a_u] = [a_v]$. Since we have $(T_i\theta'')\sigma \in \Gcal$,
      we know that $(x_u\theta'')\sigma =_\Gcal (x_v\theta'')\sigma$,
      which means $x_u\theta'' =_\Gcal x_v\theta''$. By the
      construction of $\theta''$, $x_u\theta'' =_\Gcal x_u\theta =
      a_u$ and $x_v\theta'' =_\Gcal x_v\theta = a_v$ Thus $[a_u] =
      [a_v]$ holds.


    \item
      Suppose $w$ is $x_u$ for some $1\leq u\leq m$ and $w'$ is $y_v$ for
      some $1\leq v \leq n$. We show that $x_u = y_v$ under the
      valuation $\mu'$, which means $[a_u] = [b_v]$. Since
      $(T_i\theta'')\sigma \in \Gcal$, we have $(x_u\theta'')\sigma
      =_\Gcal (y_v\theta'')\sigma$, which equals $x_u\theta'' =_\Gcal
      y_v\sigma$. Since $x_u\theta'' =_\Gcal x_u\theta = a_u$, we have
      $[a_u] = [b_v]$.

    \item
      If $w\in\{y_1,\cdots,y_n\}$ and $w'\in\{x_1,\cdots,x_m\}$, the
      case is symmetric to the above case.

    \item
      Suppose $w$ is $y_u$ and $w'$ is $y_v$ for some
      $u,v\in\{1,\cdots,n\}$. We need to show that $y_u = y_v$ under
      the valuation $\mu'$, which means $[b_u] = [b_v]$. Since we have
      $(T_i\theta'')\sigma \in \Gcal$, we know that $(y_u\theta'')\sigma
      =_\Gcal (y_v\theta'')\sigma$, which means $y_u\sigma =_\Gcal
      y_v\sigma$. Thus $[b_u] = [b_v]$ holds.
      
    \end{itemize}

  \item
    $T_i$ has the form $(w \neq w')$. This case is similar to the
    above case.

  \item
    $T_i$ has the form $(w,w'\simp w'')$.
    \begin{itemize}
    \item
      Suppose $w$ is $x_t$, $w'$ is $x_u$, $w''$ is $x_v$, for some
      $t,u,v \in\{1,\cdots,m\}$. We need to show that $R(x_t,x_u,x_v)$
      holds under the valuation $\mu'$, which is equivalent to showing
      $R([a_t],[a_u],[a_v])$. We already have that
      $(T_i\theta'')\sigma\in\Gcal$, that is, $(x_t\theta'',
      x_u\theta''\simp x_v\theta'')\in\Gcal$. By the construction of
      $\theta''$, we have $x_t\theta'' =_\Gcal x_t\theta = a_t$,
      $x_u\theta'' =_\Gcal x_u\theta = a_u$, and $x_v\theta'' =_\Gcal
      x_v\theta = a_v$. Thus the goal holds.

    \item
      If any of $w,w',w''$ is in $\{y_1,\cdots,y_n\}$, say $w$ is
      $y_u$ for some $1\leq u \leq n$, then $(y_u\theta'')\sigma = b_u$
      and $y_u$ is mapped to $[b_u]$ under $\mu'$. These cases are
      similar to the above analysis.
      
    \end{itemize}
  \end{itemize}

  This concludes the part of the proof which shows that the extended
  model $\Mcal$ is indeed a model based on a Kripke relational frame
  satisfying $\Frcal$.  We prove next that
  $\myseq{\Gcal}{\Gamma}{\Delta}$ is false in $\Mcal.$ We need to show
  the following where $\rho(m) = [m]$:
\begin{enumerate}
\item If $(a, b \simp c) \in \Gcal$ then 
$([a], [b] \simp_\Gcal [c]).$
\item If $m : A \in \Gamma$ then 
$\rho(m) \Vdash A.$
\item If $m : A \in \Delta$ then
$\rho(m) \not \Vdash A.$
\end{enumerate}
Item (1) follows from the definition of $\simp_\Gcal$. 
We prove (2) and (3) simultaneously by induction on the size of $A$. 

Base cases: when $A$ is an atomic proposition $p$.
\begin{itemize}
\item If $m:p\in\Gamma$ then $[m]\in v(p)$ by definition of $v$, so $[m] \Vdash p.$
\item Suppose $m:p\in\Delta$, but $[m] \Vdash p$. 
Then $m':p\in\Gamma$, for some $m'$ s.t. $m' =_\Gcal m$. 
This violates condition 1 in Definition~\ref{definition:hintikka_seq}. Thus $[m]\not\Vdash p$.
\end{itemize}


Inductive cases: when $A$ is a compound formula we do a case
  analysis on the main connective of $A$.
\begin{itemize}
\item If $m:A\land B \in\Gamma$, by condition 5 in
  Definition~\ref{definition:hintikka_seq}, $m:A\in\Gamma$ and $m:B\in\Gamma$. By
  the induction hypothesis, $[m]\Vdash A$ and $[m]\Vdash B$, thus
  $[m]\Vdash A\land B$.
\item If $m:A\land B \in\Delta$, by condition 6 in
  Definition~\ref{definition:hintikka_seq}, $m:A\in\Delta$ or $m:B\in\Delta$. By
  the induction hypothesis, $[m]\not\Vdash A$ or $[m]\not\Vdash B$,
  thus $[m]\not\Vdash A\land B$.
\item If $m:A\limp B\in\Gamma$, by condition 7 in
  Definition~\ref{definition:hintikka_seq}, $m:A\in\Delta$ or $m:B\in\Gamma$. By
  the induction hypothesis, $[m]\not\Vdash A$ or $[m]\Vdash B$, thus
  $[m]\Vdash A\limp B$.
\item If $m:A\limp B\in\Delta$, by condition 8 in
  Definition~\ref{definition:hintikka_seq}, $m:A\in\Gamma$ and $m:B\in\Delta$. By
  the induction hypothesis, $[m]\Vdash A$ and $[m]\not\Vdash B$, thus
  $[m]\not\Vdash A\limp B$.
\item If $m:\top^*\in\Gamma$ then  $[m] = [\epsilon]$ by condition 2 in
  Definition~\ref{definition:hintikka_seq}. Since
  $[\epsilon]\Vdash\top^*$, we obtain $[m]\Vdash\top^*$.
\item If $m:\top^*\in\Delta$, by condition 3 in
  Definition~\ref{definition:hintikka_seq}, $[m] \not = [\epsilon]$ and then
  $[m]\not\Vdash\top^*$.
\item If $m:A\mand B\in\Gamma$, by condition 9 in
  Definition~\ref{definition:hintikka_seq}, $\exists a,b,m'$ s.t. $(a,b\simp
  m')\in\Gcal$ and $[m] =[m']$ and $a:A\in\Gamma$ and
  $b:B\in\Gamma$. By the induction hypothesis, $[a]\Vdash A$ and
  $[b]\Vdash B$. Thus $[a],[b]\simp_{\Gcal} [m]$ holds and $[m]\Vdash
  A\mand B$.
\item If $m:A\mand B\in\Delta$, by condition 10 in
  Definition~\ref{definition:hintikka_seq}, $\forall a,b,m'$ if $(a,b\simp
  m')\in\Gcal$ and $[m]=[m']$, then $a:A\in\Delta$ or
  $b:B\in\Delta$. By the induction hypothesis, if such $a,b$ exist,
  then $[a]\not\Vdash A$ or $[b]\not\Vdash B$. For any
  $[a],[b]\simp_{\Gcal} [m]$, there must be some $(a',b'\simp
  m'')\in\Gcal$ s.t. $[a] = [a'], [b] = [b'], [m] = [m'']$. Then
  $[a]\not\Vdash A$ or $[b]\not\Vdash B$ therefore
  $[m]\not\Vdash A\mand B$.
\item If $m:A\mimp B\in\Gamma$, by condition 11 in
  Definition~\ref{definition:hintikka_seq}, $\forall a,b,m'$ if $(a,m'\simp
  b)\in\Gcal$ and $[m]=[m']$, then $a:A\in\Delta$ or
  $b:B\in\Gamma$. By the induction hypothesis, if such $a,b$ exists,
  then $[a]\not\Vdash A$ or $[b]\Vdash B$. Consider any
  $[a],[m]\simp_{\Gcal} [b]$. There must be some $(a',m''\simp
  b')\in\Gcal$ s.t. $[a] = [a']$, $[m''] = [m]$, and $[b] = [b']$. So $[a] \not\Vdash A$ or $[b]\Vdash B$,
  thus $[m]\Vdash A\mimp B$.
\item If $m:A\mimp B\in\Delta$, by condition 12 in
  Definition~\ref{definition:hintikka_seq}, $\exists a,b,m'$ s.t. $(a,m'\simp
  b)\in\Gcal$ and $[m]=[m']$ and $a:A\in\Gamma$ and $b:B\in\Delta$. By
  the induction hypothesis, $[a]\Vdash A$ and $[b]\not\Vdash B$ and
  $[a],[m]\simp_{\Gcal} [b]$ holds, thus $[m]\not\Vdash A\mimp B$. 
\end{itemize}

This concludes the proof.
\end{proof}

To prove the completeness of an arbitrary finite instance of $\lsg$,
we have to show that any given unprovable sequent can be extended to a
Hintikka sequent with the corresponding conditions. To do so we need a
way to enumerate all possible backwards applicable rules in a fair way
so that every rule will be chosen infinitely often.  Traditionally,
this is achieved via a fair enumeration strategy of every principal
formula of every rule.  Since our calculus may contain structural
rules with no principal formulae, we need to include them in the
enumeration strategy as well. For this purpose, we define a notion of
{\em extended formulae}, given by the grammar:
\begin{center}
$ExF ::= F~|~\EMbb~|~\GSbb_1~|\cdots|~\GSbb_q$
\end{center}
where $F$ is a formula, and $\GSbb_1,\cdots,\GSbb_q$ are ``dummy'' constant
principal formulae corresponding to each structural rule respectively, and
$\EMbb$ is for the rule $EM$.

Let $q$ be the number of
structural rules, synthesised from $q$ frame axioms. In those frame
axioms, let $k_{max}$ be the largest number of relational atoms on the left hand side
(i.e., $S$s in the general axiom), and $n_{max}$ be the largest number of
existentially quantified variables (i.e., $y$'s in the general axiom).
A scheduler enumerates each combination of left or right of
turnstile, a label, an extended formula and at most two relational
atoms infinitely often.

\begin{definition}[Scheduler $\phi$]
\label{definition:fair}
A {\em schedule} is a tuple $(O, m, \exfml{F}, \Rcal)$, where $O$ is
either $0$ (left) or $1$ (right), $m$ is a label, $\exfml{F}$ is an
extended formula and $\Rcal$ is a set of relational atoms such that $|\Rcal|
\leq k_{max}$. Let $\Scal$
denote the set of all schedules.  A {\em scheduler} is a function from
natural numbers $\Ncal$ to $\Scal.$ 
A scheduler $\phi$ is {\em fair} if 
for every schedule $S$, the set $\{i \mid \phi(i) = S \}$ is infinite.
\end{definition}

\begin{lemma}
There exists a fair scheduler.
\end{lemma}
\begin{proof}
Our proof is similar to the proof of \textit{fair strategy} of 
Larchey-Wendling~\cite{wendling2012}. To adapt their proof, 
we need to show
that the set $\Scal$ is countable. 
This follows from the fact that $\Scal$ is a finite product of countable sets. 
\end{proof}
From now on, we shall fix a fair scheduler, which we call 
$\phi$. 
We assume that the set $\Lcal$ of labels 
is totally ordered, and its elements can be enumerated
as $a_0,a_1,a_2,\ldots$ where $a_0 = \epsilon.$
This indexing 
is used
to select fresh labels in our construction of Hintikka sequents.

We say the formula $F$ is not cut-free provable in a finite instance
of $\lsg$ if for an arbitrary label $w \not = \epsilon$,
the sequent $\vdash w : F$ is not cut-free derivable in
that instance of $\lsg$. Since we
shall be concerned only with cut-free provability, in the following
when we mention derivation, we mean cut-free derivation.

For a structural rule obtained from a frame axiom of the usual
form~\eqref{eq:ax} we say the structural rule is \emph{backwards
  applicable} on a sequent $\myseq{\Gcal}{\Gamma}{\Delta}$ iff there
is a set $\Gcal'\subseteq\Gcal$ of non-equality relational atoms that
matches the schema $S_1,\cdots,S_k$, and $E(\Gcal)\vdash s_1 =
t_1,\cdots,E(\Gcal)\vdash s_p = t_p$ holds.

\begin{definition}
\label{definition:sequent-series}
Let $F$ be a formula which is not provable in an instance of $\lsg$.
We construct a series of finite sequents
$\langle \myseq{\Gcal_i}{\Gamma_i}{\Delta_i} \rangle_{i \in \Ncal}$
from $F$ where $\Gcal_1 = \Gamma_1 = \emptyset$ and $\Delta_1 =
a_1:F$.

Assuming that $\myseq{\Gcal_i}{\Gamma_i}{\Delta_i}$ has been defined,
we define $\myseq{\Gcal_{i+1}}{\Gamma_{i+1}}{\Delta_{i+1}}$ as
follows.
Suppose $\phi(i) = (O_i,m_i,\exfml{F}_i,\Rcal_i).$ Recall that $n_{max}$ is
the maximum number of existentially quantified variables in the frame
axioms, take $\alpha$ as $max(n_{max},2)$.

\begin{itemize}

\item Case $O_i = 0$, $\exfml{F}_i$ is a $\psl$-formula $C_i$ and
  $m_i:C_i\in\Gamma_i$:
\begin{itemize}
\item $C_i = F_1\land F_2$: then $\Gcal_{i+1} = \Gcal_i$,
  $\Gamma_{i+1} = \Gamma_i\cup\{m_i:F_1,m_i:F_2\}$, $\Delta_{i+1} =
  \Delta_i$.
\item $C_i = F_1\limp F_2$: if there is no derivation for
  $\myseq{\Gcal_i}{\Gamma_i}{m_i:F_1;\Delta_i}$ then $\Gamma_{i+1} =
  \Gamma_i$, $\Delta_{i+1} = \Delta_i\cup\{m_i:F_1\}$. Otherwise
  $\Gamma_{i+1} = \Gamma_i\cup\{m_i:F_2\}$, $\Delta_{i+1} =
  \Delta_i$. In both cases, $\Gcal_{i+1} = \Gcal_i$.
\item $C_i = \top^*$: then
  $\Gcal_{i+1} = \Gcal_i\cup\{(m_i= \epsilon)\}$,
  $\Gamma_{i+1} =
  \Gamma_i$, $\Delta_{i+1} = \Delta_i$.
\item $C_i = F_1\mand F_2$: then $\Gcal_{i+1} =
  \Gcal_i\cup\{(a_{\alpha i},a_{\alpha i+1}\simp m_i)\}$ and
  $\Gamma_{i+1} = \Gamma_i\cup\{a_{\alpha i}:F_1,a_{\alpha i+1}:F_2\}$,
  where $a_{\alpha i},a_{\alpha i+1}$ are fresh labels, and
  $\Delta_{i+1} = \Delta_i$.
\item $C_i = F_1\mimp F_2$ and $\Rcal_i = \{(x,m\simp  y)\}\subseteq\Gcal_i$ and $E(\Gcal_i)\vdash (m = m_i)$:
  if $\myseq{\Gcal_i}{\Gamma_i}{x:F_1;\Delta_i}$ has no derivation, then
  $\Gamma_{i+1} = \Gamma_i$, $\Delta_{i+1} =
  \Delta_i\cup\{x:F_1\}$. Otherwise $\Gamma_{i+1} =
  \Gamma_i\cup\{y:F_2\}$, $\Delta_{i+1} = \Delta_i$. In both cases,
  $\Gcal_{i+1} = \Gcal_i$.
\end{itemize}

\item Case $O_i = 1$, $\exfml{F}_i$ is a $\psl$-formula $C_i$, and
  $m_i:C_i\in\Delta$:
\begin{itemize}
\item $C_i = F_1\land F_2$: if there is no derivation for
  $\myseq{\Gcal_i}{\Gamma_i}{m_i:F_1;\Delta_i}$ then $\Delta_{i+1} =
  \Delta_i\cup\{m_i:F_1\}$. Otherwise $\Delta_{i+1} =
  \Delta_i\cup\{m_i:F_2\}$. In both cases, $\Gcal_{i+1} = \Gcal_i$ and
  $\Gamma_{i+1} = \Gamma_i$.
\item $C_i = F_1\limp F_2$: then $\Gamma_{i+1} =
  \Gamma\cup\{m_i:F_1\}$, $\Delta_{i+1} = \Delta_i\cup\{m_i:F_2\}$,
  and $\Gcal_{i+1} = \Gcal_i$.
\item $C_i = F_1\mand F_2$ and $\Rcal_i = \{(x,y\simp
  m)\}\subseteq\Gcal_i$ and $E(\Gcal_i)\vdash (m_i = m)$: if
  $\myseq{\Gcal_i}{\Gamma_i}{x:F_1;\Delta_i}$ has no derivation, then
  $\Delta_{i+1} = \Delta_i\cup\{x:F_1\}$. Otherwise $\Delta_{i+1} =
  \Delta_i\cup\{y:F_2\}$. In both cases, $\Gcal_{i+1} = \Gcal_i$ and
  $\Gamma_{i+1} = \Gamma_i$.
\item $C_i = F_1\mimp F_2$: then $\Gcal_{i+1} =
  \Gcal_i\cup\{(a_{\alpha i},m_i\simp a_{\alpha i+1})\}$, $\Gamma_{i+1} =
  \Gamma_i\cup\{a_{\alpha i}:F_1\}$, 
   where $a_{\alpha i},a_{\alpha i+1}$ are fresh labels, and
  $\Delta_{i+1} = \Delta_i\cup\{a_{\alpha i+1}:F_2\}$.
\end{itemize}

\item Case $\exfml{F}_i = \EMbb$, and $\Rcal_i = \{(w,w'\simp w'')\}$ where
  $w,w'\in\{a_0,\cdots,a_{\alpha i + \alpha -1}\}$. If there is no derivation for $\myseq{w = w';\Gcal_i}{\Gamma_i}{\Delta_i}$, then $\Gcal_{i+1} = \Gcal_i\cup\{w = w'\}$. Otherwise $\Gcal_{i+1} = \Gcal_i\cup\{w\neq w'\}$. In both cases, $\Gamma_{i+1} = \Gamma_i$ and $\Delta_{i+1} = \Delta_i$.

\item Case $\exfml{F}_i \in\{\GSbb_1,\cdots,\GSbb_q\}$. Assume
  without loss of generality that $\exfml{F}_i$ is $\GSbb_j$ for some $1\leq j \leq q$
  and $\GSbb_j$ represents a structural rule that corresponds to the
  frame axiom
    \[
    \forall x_1,\dots,x_m. (s_1 = t_1 \,\& \cdots \&\, s_p = t_p 
   \,\&\, S_1 \,\& \cdots \&\, S_k \,\Rightarrow\, \exists y_1, \dots, y_n. (T_1 \,\& \cdots \&\, T_l))
   \]
  Suppose we can find some substitution $\theta$ such that
  $\bigcup_{1\leq i \leq k} \{S_i\theta\} = \Rcal_i \subseteq
  \Gcal_i$. Also suppose that for each $1\leq u \leq p$,
  $E(\Gcal_i)\vdash s_u\theta = t_u\theta$. We create $n$ fresh labels
  $a_{\alpha i},a_{\alpha i+1},\cdots,a_{\alpha i+n-1}$ and a
  substitution $\sigma = [a_{\alpha i}/y_1,\cdots,a_{\alpha
      i+n-1}/y_n]$. Let $\Gcal_{i+1} =
  \Gcal_i\cup\{(T_1\theta)\sigma,\cdots,(T_{l}\theta)\sigma\}$,
  $\Gamma_{i+1} = \Gamma_i$ and $\Delta_{i+1} = \Delta_i$.  Note that
  finding such a $\theta$ is computable since $S_i$ and $\Rcal_i$ are
  given. This is a simple case of the ACUI unification
  problem~\cite{Robinson2001}.
  
\item In all other cases, $\Gcal_{i+1} = \Gcal_i$, $\Gamma_{i+1} =
  \Gamma_i$ and $\Delta_{i+1} = \Delta_i$.
\end{itemize}
\end{definition}

Intuitively, each tuple $(O_i, m_i, \exfml{F_i}, \Rcal_i)$ corresponds
to a potential (backwards)  rule application. If the components of the rule
application are in the current sequent, we apply the corresponding
rule to these components to obtain the new premises.  The indexing of labels guarantees that the
choice of $a_{\alpha i},\cdots,a_{\alpha i+\alpha-1}$ are always fresh
for the sequent $\myseq{\Gcal_i}{\Gamma_i}{\Delta_i}$.  The
construction in Definition~\ref{definition:sequent-series}
non-trivially extends a similar construction of Hintikka Constrained
Set of Statements due to~\cite{wendling2012}, in addition to which we
have to consider the cases for structural rules.

We say $\myseq{\Gcal'}{\Gamma'}{\Delta'} \subseteq
\myseq{\Gcal}{\Gamma}{\Delta}$ iff $\Gcal'\subseteq\Gcal$,
$\Gamma'\subseteq\Gamma$ and $\Delta'\subseteq\Delta$. A labelled
sequent $\myseq{\Gcal}{\Gamma}{\Delta}$ is \textit{finite} if
$\Gcal,\Gamma,\Delta$ are finite sets. Define $\myseq{\Gcal'}{\Gamma'}{\Delta'}
\subseteq_f \myseq{\Gcal}{\Gamma}{\Delta}$ iff
$\myseq{\Gcal'}{\Gamma'}{\Delta'} \subseteq \myseq{\Gcal}{\Gamma}{\Delta}$ and
$\myseq{\Gcal'}{\Gamma'}{\Delta'}$ is finite. If
$\myseq{\Gcal}{\Gamma}{\Delta}$ is a finite sequent, it is
\textit{non-provable} iff it does not have a derivation in
the instance of $\lsg$.
A (possibly infinite) sequent $\myseq{\Gcal}{\Gamma}{\Delta}$
is \textit{finitely non-provable} iff every
$\myseq{\Gcal'}{\Gamma'}{\Delta'}\subseteq_f \myseq{\Gcal}{\Gamma}{\Delta}$ is
non-provable.


We write $\Lcal_i$ for the set of labels occurring in the sequent
$\myseq{\Gcal_i}{\Gamma_i}{\Delta_i}$. Thus $\Lcal_1 = \{a_1\}$. 
The following lemma states some properties of the construction of 
the sequents $\myseq{\Gcal_i}{\Gamma_i}{\Delta_i}$.
\begin{lemma}
\label{lm:construction}
For any $i\in\mathcal{N}$, the following properties hold:
\begin{enumerate}
\item $\myseq{\Gcal_i}{\Gamma_i}{\Delta_i}$ is non-provable;
\item $\Lcal_i\subseteq \{a_0, a_1,\cdots,a_{\alpha i-1}\}$;
\item $\myseq{\Gcal_i}{\Gamma_i}{\Delta_i}\subseteq
  \myseq{\Gcal_{i+1}}{\Gamma_{i+1}}{\Delta_{i+1}}$.
\end{enumerate}
\end{lemma}
\begin{proof}
Item 1 is based on the fact that the inference rules preserve
falsifiability upwards, and we always choose the branch with no
derivation.  To show item 2, we do an induction on $i$. Base case, $i = 1$, 
$\Lcal_1 \subseteq \{a_0, a_1\}$ (recall that $a_0 = \epsilon$).  Inductive cases: suppose item 2 holds for any
$i \leq n$, for $n+1$, we consider four cases depending on which rule
is applied on $\myseq{\Gcal_i}{\Gamma_i}{\Delta_i}$.
\begin{itemize}
\item If $\mand L$ is applied, then $\Lcal_{i+1} =
  \Lcal_i\cup\{a_{\alpha i},a_{\alpha
    i+1}\}\subseteq\{a_1,\cdots,a_{\alpha i+\alpha -1}\}$ since
  $\alpha \leq 2$.
\item If $\mimp R$ is applied, same as above.
\item If a structural rule $\GSbb_j$ is applied, and it generates $n'$ fresh labels.
  By construction, $n'\leq \alpha$,
  thus $\Lcal_{i+1} = \Lcal_i\cup\{a_{\alpha i}, a_{\alpha i + 1},\cdots, a_{\alpha i + n' -1}\}\subseteq\{a_1,\cdots,a_{\alpha i+ \alpha -1}\}$.
\item Otherwise, $\Lcal_{i+1} = \Lcal_i \subseteq \{a_1,\cdots,a_{2i+1}\}$.   
\end{itemize}
Item 3 is obvious from the construction of $\myseq{\Gcal_{i+1}}{\Gamma_{i+1}}{\Delta_{i+1}}.$
\end{proof}

Given the construction of the series of sequents we have just seen above, we define
a notion of a {\em limit sequent}, as the union of 
every sequent in the series.

\begin{definition}[Limit sequent]
\label{definition:lim_seq} 
Let $F$ be a formula unprovable in the instance of $\lsg.$
The {\em limit sequent for $F$} is the sequent 
$\myseq{\Gcal^\omega}{\Gamma^\omega}{\Delta^\omega}$ 
where 
$\Gcal^\omega = \bigcup_{i\in\mathcal{N}}\Gcal_i$,
$\Gamma^\omega = \bigcup_{i\in\mathcal{N}}\Gamma_i$,
and
$\Delta^\omega = \bigcup_{i\in\mathcal{N}}\Delta_i$.
\end{definition}

The following lemma shows that the limit sequent defined above is
indeed an instance of the general Hintikka sequent, thus we can use it to extract a
counter-model.

\begin{lemma}
\label{lem:lim_hintikka}
If $F$ is a formula unprovable in the instance of $\lsg$, then the
limit labelled sequent for $F$ is the instance of the general Hintikka
sequent with the corresponding conditions.
\end{lemma}

\begin{proof}
Let $\myseq{\Gcal^\omega}{\Gamma^\omega}{\Delta^\omega}$ be the limit sequent. 
First we show that $\myseq{\Gcal^\omega}{\Gamma^\omega}{\Delta^\omega}$ is finitely non-provable. Consider any
$\myseq{\Gcal}{\Gamma}{\Delta}\subseteq_f \myseq{\Gcal^\omega}{\Gamma^\omega}{\Delta^\omega}$,
we show that $\myseq{\Gcal}{\Gamma}{\Delta}$ has no derivation. Since
$\Gcal,\Gamma,\Delta$ are finite sets, there exists $i\in
\mathcal{N}$ s.t. $\Gcal\subseteq \Gcal_i$, $\Gamma\subseteq
\Gamma_i$, and $\Delta\subseteq \Delta_i$. 
Moreover, from Lemma~\ref{lm:construction} Item 1, $\myseq{\Gcal_i}{\Gamma_i}{\Delta_i}$ is not provable in $\lsg$.
Since weakening is admissible in $\lsg$,
$\myseq{\Gcal}{\Gamma}{\Delta}\subseteq_f
\myseq{\Gcal_i}{\Gamma_i}{\Delta_i}$ cannot be provable either.
So condition 1, 2, 4, 13 in Definition~\ref{definition:hintikka_seq} hold for the limit sequent, 
for otherwise we would be able to construct a provable finite labelled sequent from 
the limit sequent. We show the proofs that the other conditions 
in~Definition~\ref{definition:hintikka_seq} are also satisfied by the
limit sequent. The following cases are numbered according to items in Definition~\ref{definition:hintikka_seq}.
\begin{enumerate}
\setcounter{enumi}{2}

\item If $m:\top^*\in\Gamma^\omega$, then $m:\top^*\in\Gamma_i$, for some
  $i\in\mathcal{N}$, since each labelled formula from $\Gamma^\omega$ must appear somewhere
  in the sequence.
Then there exists $j > i$ such that $\phi(j) = (0,m,\top^*,R)$ where
this formula becomes principal.
By construction, $(m= \epsilon)\in\Gcal_{j+1}\subseteq\Gcal^\omega$. 
So we deduce that $m =_{\Gcal^\omega} \epsilon$.

\setcounter{enumi}{4}

\item If $m:F_1\land F_2\in\Gamma^\omega$, then it is in some 
$\Gamma_i$, where $i\in\mathcal{N}$. Since $\phi$ select the formula infinitely often, 
there is $j > i$ such that $\phi(j) = (0,m,F_1\land F_2,R)$. Then by construction 
$\{m:F_1,m:F_2\}\subseteq \Gamma_{j+1}\subseteq \Gamma^\omega$. 

\item If $m:F_1\land F_2\in\Delta^\omega$, then it is in some $\Delta_i$, where $i\in\mathcal{N}$. 
Since $\phi$ select the formula infinitely often, there is $j > i$ such that
$\phi(j) = (1,m,F_1\land F_2,R)$. Then by construction $m:F_n\in\Delta_{j+1}\subseteq\Delta^\omega$, 
where $n\in\{1,2\}$ and $\myseq{\Gcal_j}{\Gamma_j}{m:F_n;\Delta_j}$ does not have a derivation.

\item If $m:F_1\limp F_2\in\Gamma^\omega$, similar to case 3.

\item If $m:F_1\limp F_2\in\Delta^\omega$, similar to case 2.

\item\label{item:1} If $m:F_1\mand F_2\in\Gamma^\omega$, then it is in some $\Gamma_i$, where $i\in\mathcal{N}$. 
Then there exists $j > i$ such that $\phi(j) = (0,m,F_1\mand F_2,R)$. 
By construction $\Gcal_{j+1} = \Gcal_j\cup\{(a_{2j},a_{2j+1}\simp m)\}\subseteq\Gcal^\omega$, 
and $\Gamma_{j+1} = \Gamma_j\cup\{a_{2j}:F_1,a_{2j+1}:F_2\}\subseteq\Gamma^\omega$. 

\item\label{item:2} If $m:F_1\mand F_2\in\Delta^\omega$, then it is in some $\Delta_i$, where 
$i\in\mathcal{N}$. For any $(x,y\simp m')\in\Gcal^\omega$ such that $m =_{\Gcal^\omega} m'$, 
there exists $j > i$ such that $(x,y\simp m')\in\Gcal_j$ and $m =_{\Gcal_j} m'$. 
Also, there exists $k > j$ such that $\phi(k) = (1,m,F_1\mand
F_2,\{(x,y\simp m')\})$ where the labelled formula becomes principal.
Since $(x,y\simp m')\in\Gcal_k$ and $m =_{\Gcal_k} m'$,
we have either $x : F_1 \in \Delta_{k+1} \subseteq \Delta^\omega$ 
or $y : F_2 \in \Delta_{k+1} \subseteq \Delta^\omega.$

\item If $m:F_1\mimp F_2\in\Gamma^\omega$, similar to case~\ref{item:1}.

\item If $m:F_1\mimp F_2\in\Delta^\omega$, similar to case~\ref{item:2}.





\setcounter{enumi}{13}
  
\item
  For each pair $a_p,a_q\in\Lcal$, assume
  without loss of generality
  that $p\geq q$. Then there is some natural number $j\geq q$
  such that $\phi(j) = (O,m,\EMbb,\{(a_p,a_q\simp m')\})$. Then either (1)
  $\Gcal_{j+1} = \Gcal_j\cup\{a_p = a_q\}$, or (2) $\Gcal_{j+1} =
  \Gcal_j\cup\{a_p\neq a_q\}$, depending on which choice gives a
  finitely non-provable sequent
  $\myseq{\Gcal_{j+1}}{\Gamma_{j+1}}{\Delta_{j+1}}$. If (1) holds,
  then $a_p = a_q\in \Gcal_{j+1}\subseteq \Gcal^\omega$, and obviously
  $E(\Gcal^w) \vdash a_p = a_q$, thus $a_p=_{\Gcal^\omega} a_q$. If
  (2) holds, then $a_p\neq a_q\in\Gcal_{j+1}\subseteq\Gcal^\omega$.

\item For an arbitrary instance $\GSbb_n$ of condition 15, assume
  without loss of generality that this instance is of the form
  \[
    \forall x_1,\dots,x_m. (s_1 = t_1 \,\& \cdots \&\, s_p = t_p 
   \,\&\, S_1 \,\& \cdots \&\, S_k \,\Rightarrow\, \exists y_1, \dots, y_n. (T_1 \,\& \cdots \&\, T_l))
  \]
  Suppose there is a substitution $\theta$ such that for each $1\leq i
  \leq k$, $S_i\theta\in\Gcal^\omega$ and for each $1\leq i\leq p$,
  $s_i =_{\Gcal^\omega} t_i$. Then there must be some natural number
  $j$ such that $\bigcup_{1\leq i \leq k} S_i\theta\subseteq \Gcal_j$
  and for each $1\leq i\leq p, E(\Gcal_j)\vdash s_i = t_i$. There also
  exists $k\geq j$ such that $\phi(k) = (O,m,\GSbb_n,\Rcal)$ where
  $\bigcup_{1\leq i \leq k} S_i\theta = \Rcal\subseteq \Gcal_k$. It is
  obvious that for each $1\leq i\leq p, E(\Gcal_k)\vdash s_i =
  t_i$. By construction, there is some substitution $\sigma$ such that
  $\Gcal_{k+1} = \bigcup_{1\leq i\leq l} (T_i\theta)\sigma \cup
  \Gcal_k$. Therefore the corresponding instance of condition (15)
  holds.
  
  
\end{enumerate}
The above covers all the conditions in Definition~\ref{definition:hintikka_seq}.
\end{proof}


Finally we can state the completeness theorem: whenever a formula has
no derivation in the instance of $\lsg$,  we can extract an infinite
counter-model based on the limit sequent and the Kripke relational
frame with the corresponding conditions.

\begin{theorem}
\label{thm:completeness_lsg}
Every formula $F$ unprovable in the instance of $\lsg$ is not valid in
the Kripke relational models with the corresponding conditions.
\end{theorem}
\begin{proof}
We construct a limit sequent $\myseq{\Gcal^\omega}{\Gamma^\omega}{\Delta^\omega}$ for $F$
following Definition~\ref{definition:lim_seq}. Note that by the construction of the limit sequent,
we have $a_1 : F \in \Delta^\omega.$ By Lemma~\ref{lem:lim_hintikka},
this limit sequent is a Hintikka sequent, and therefore by Lemma~\ref{lem:hintikka_sat},
$\myseq{\Gcal^\omega}{\Gamma^\omega}{\Delta^\omega}$ is falsifiable.
This means there exists a model $(\Fcal,\val,\rho)$ that satisfies $\Gcal^\omega$ and $\Gamma^\omega$
and falsifies every element of $\Delta^\omega$, including $a_1 : F$, which means that
$F$ is false at world $\rho(a_1).$  Thus $F$ is not valid.
\end{proof}


\begin{corollary}[Completeness]
\label{thm:completeness_pasl}
Every formula $F$ unprovable in $\lspslh$ is not valid in $\psl$, or
contrapositively, if $F$ is PASL-valid then $F$ is provable in  $\lspslh$.
\end{corollary}




%% file: extension_psl.tex
\section{Extensions of $\psl$}
\label{sec:extension_psl}

In Section~\ref{subsec:lspsl} we presented our proof theory
as a \emph{family} of sequent calculi $\lsg$ parameterised by a
choice of \emph{frame axioms}. This section establishes the benefit
of this modular approach by showing how the additional axioms for
separation algebras proposed by~\cite{dockins2009}, namely
\emph{indivisible unit}, \emph{disjointness}, \emph{splittability},
and \emph{cross-split}, can be accommodated within our approach.
We also look at the generalisation of the definition of separation
algebra given by rejecting \emph{cancellativity}. We further discuss
which of these axioms manifest in the various examples of separation algebra surveyed in Section~\ref{subsec:pasl},
along with two further non-cancellative examples, helping to establish
how different abstract semantics correspond to different
concrete semantics, as summarised in Table~\ref{tab:models}.

\begin{table}
\centering
{\small
\begin{tabular}{|l|c|c|c|c|c|}
\hline
Concrete model & Indivisible unit & Disjointness & Splittability & Cross-Split & Cancellativity \\
\hline
Heaps & \checkmark & \checkmark & $\mathbb{\times}$ & \checkmark & \checkmark  \\
Fractional permissions & \checkmark & $\mathbb{\times}$ & \checkmark & \checkmark & \checkmark  \\ 
Named permissions & \checkmark & \checkmark & $\mathbb{\times}$ & \checkmark & \checkmark \\
Counting permissions & \checkmark & $\mathbb{\times}$ & $\mathbb{\times}$ & \checkmark & \checkmark  \\
Binary trees & \checkmark & \checkmark & \checkmark & \checkmark & \checkmark \\
Heaps (finite locations) & \checkmark & \checkmark & $\mathbb{\times}$ & \checkmark & \checkmark  \\
Petri nets & \checkmark & $\mathbb{\times}$ & $\mathbb{\times}$ & \checkmark & \checkmark \\
Petri nets with capacity $1$ & \checkmark & \checkmark & $\mathbb{\times}$ & \checkmark & \checkmark \\
Endpoint heaps & \checkmark & \checkmark & $\mathbb{\times}$ & \checkmark & \checkmark  \\
Monotonic counter & \checkmark & $\mathbb{\times}$ & \checkmark & \checkmark & $\mathbb{\times}$ \\
Logical heaps & \checkmark & $\mathbb{\times}$ & \checkmark & \checkmark & $\mathbb{\times}$ \\
\hline
\end{tabular}
}
\caption{Some concrete separation algebras and their abstract properties}
\label{tab:models}
\end{table}

It is an immediate corollary of the soundness
(Theorem~\ref{thm:soundness}) and completeness
(Theorem~\ref{thm:completeness_lsg}) of the general proof system
$\lsg$ that we have sound and complete proof rules for any combination
of these axioms.

By focusing on the properties emphasised by~\cite{dockins2009} we do
not mean to imply that these constitute a canonical list of abstract
properties worth considering. Indeed, an advantage of our modular
setting is that it will allow investigations of other abstract properties. For
example, the property below expresses that any element can be non-trivially extended:
\[
  \forall h_1.\exists h_2,h_3.~ h_2\neq\epsilon ~\&~ R(h_1,h_2,h_3)
\]
The property is satisfied by all models of Table~\ref{tab:models}
except heaps with finitely many locations~\cite{Jensen:High}, and
Petri nets with finite capacity and enough tokens to fill all
places. This axiom is in frame axiom form so we may synthesise the
rule below with the side-condition that $y$ and $z$ may not appear in the
conclusion:
\begin{center}
  \AxiomC{$\myseq{\Gcal;y\neq\epsilon;(x,y\simp z)}{\Gamma}{\Delta}$}
  \UnaryInfC{$\myseq{\Gcal}{\Gamma}{\Delta}$}
  \DisplayProof
\end{center}
With this rule we may prove the formula $\lnot(\lnot\top^{\mand}
\mimp\bot)$ which would otherwise be unprovable.

\subsection{Indivisible unit} 
The unit $\epsilon$ in a commutative monoid 
$(H,\circ,\epsilon)$ is \emph{indivisible} iff the following holds for 
any $h_1,h_2\in H$: 
$$\mbox{ if } h_1\circ h_2 = \epsilon \mbox{ then } h_1 = \epsilon.$$
Relationally, this corresponds to the first-order condition: 
$$\forall h_1,h_2\in H. \mbox{ if } R(h_1, h_2, \epsilon) \mbox{ then
} h_1 = \epsilon.$$
This also implies that $h_2 = \epsilon$ whenever
$h_1 \circ h_2 = \epsilon.$ Indivisible unit can be axiomatised by the
formula $\top^*\land (A\mand B)\limp A$~\cite{brotherston2013}.

To synthesise a proof rule for indivisible unit we first transform the
above condition to the following form:
$$\forall h_1,h_2,h_0.~ h_0 =
\epsilon ~\&~ R(h_1, h_2, h_0) ~\Rightarrow~ h_1 = \epsilon.$$
It is easy to check that this
form satisfies the frame axiom conditions. The corresponding
structural rule is:
\begin{center}
  \AxiomC{$\myseq{\Gcal;(x,y\simp z);x = \epsilon}{\Gamma}{\Delta}$}
  \AxiomC{$E(\Gcal)\vdash z =\epsilon$}
  \RightLabel{IU}
  \BinaryInfC{$\myseq{\Gcal;(x,y\simp z)}{\Gamma}{\Delta}$}
  \DisplayProof
\end{center}
We can also deduce that $E(\Gcal)\vdash y = \epsilon$. The following is easy to confirm:
\begin{proposition}
\label{prop:iu_axiom}
The formula $\top^*\land (A\mand B)\limp A$ is provable in $\lspslh +
IU$.
\end{proposition}

\begin{example}
It is trivial to confirm that all the concrete separation algebras surveyed in
Section~\ref{subsec:pasl} satisfy the indivisible unit axiom;
we are not aware of any
separation algebras with applications to program verification that fail to do so.
\end{example}

\subsection{Disjointness}
The separating conjunction $\mand$ in Reynolds's separation logic
requires that the two combined heaps have disjoint
domains~\cite{reynolds2002}. Without concrete semantics that give
meaning to the ``points-to'' predicate
$\mapsto$ we cannot express this notion of disjointness. However,
abstract semantics do allow us to discuss a special case where we
try to combine a non-empty heap with itself.
In a separation algebra $(H, \circ, \epsilon)$,
\emph{disjointness} is defined by the following additional
requirement:
$$\forall h_1,h_2\in H. \mbox{ if } h_1\circ h_1 = h_2 \mbox{ then }
h_1 = \epsilon.$$

The above can be expressed relationally:
$$\forall h_1,h_2\in H. \mbox{ if } R(h_1, h_1, h_2) \mbox{ then } h_1
= \epsilon.$$

To create a structural rule for it, we first need to convert the above
condition into
$$\forall h_1,h_2,h_3.~ h_1 = h_3 ~\&~ R(h_1, h_3,
h_2) ~\Rightarrow~ h_1 = \epsilon.$$

Then we obtain the structural rule for disjointness as below.
\begin{center}
  \AxiomC{$\myseq{\Gcal;(x,y\simp z);x=\epsilon}{\Gamma}{\Delta}$}
  \AxiomC{$E(\Gcal)\vdash x = y$}
  \RightLabel{$D$}
  \BinaryInfC{$\myseq{\Gcal;(x,y\simp z)}{\Gamma}{\Delta}$}
  \DisplayProof
\end{center}
Note that $E(\Gcal)\vdash z = \epsilon$ is a direct consequence.

Disjointness implies indivisible unit (but not vice versa), 
as shown by~\cite{dockins2009}. We can prove the axiom for 
indivisible unit by using $\lspslh + D$:

\begin{proposition}
\label{lem:d_iu_axiom}
The formula $\top^*\land (A\mand B)\limp A$ is provable in $\lspslh +
D.$
\end{proposition}

\begin{example}
In the cases of Example~\ref{ex:partcomsemi}, where separation algebras
are defined via a partial commutative semigroup $(V,\star)$, the disjointness
property holds iff there exist no $v\in V$ such that $v\star v$ is defined. This is
the case for heaps, named permissions, and the binary tree share model. On
the other hand, disjointness fails to hold for fractional permissions (where $(v,i)\star
(v,i)$ is defined so long as $i\leq 0.5$) and counting permissions (for example
$(v,-1)\star(v,-1)=(v,-2)$).

Disjointness fails in general for markings of Petri nets, as a marking can be
combined with itself by doubling its number of tokens at all places. However disjointness
holds in the presence of a global capacity constraint $\kappa=1$.
\end{example}

\subsection{Splittability}
The property of infinite splittability is useful when
reasoning about the kinds of resource sharing that occur in
divide-and-conquer style computations~\cite{dockins2009}. A separation
algebra $(H,\circ,\epsilon)$ has \emph{splittability} if
$$\forall h_0 \in H\setminus\{\epsilon\}, \exists h_1, h_2 \in H\setminus \{\epsilon\}
\mbox{ such that } h_1\circ h_2 = h_0.$$ 
Relationally, this corresponds to:
$$\forall h_0.~ h_0 \neq \epsilon ~\Rightarrow~ \exists h_1,h_2. h_1\neq \epsilon ~\&~ h_2\neq\epsilon ~\&~ R(h_1, h_2, h_0).$$

Splittability can be axiomatised as the formula $\lnot\top^* \limp
(\lnot\top^* \mand \lnot \top^*)$~\cite{brotherston2013}. We give the
following structural rule for this property
with the side-condition that $x,y$ do not occur in the conclusion:
\begin{center}
  \AxiomC{$\myseq{\Gcal;(z\neq \epsilon);(x,y\simp
      z);x\neq\epsilon;y\neq\epsilon}{\Gamma}{\Delta}$}
  \RightLabel{$S$}
  \UnaryInfC{$\myseq{\Gcal;(z\neq \epsilon)}{\Gamma}{\Delta}$}
  \DisplayProof
\end{center}

\begin{proposition}
\label{prop:add_s}
The axiom $\lnot\top^* \limp (\lnot\top^* \mand \lnot \top^*)$ for
splittability is provable in $\lspslh + S$.
\end{proposition}

\begin{example}
In the case of separation algebras defined via a partial commutative semigroup
$(V,\star)$, splittability holds iff all $v\in V$ are in the image of $\star$. This
holds for fractional permissions, as each $(v,i)$ is $(v,\frac{i}{2})\star(v,
\frac{i}{2})$. The binary tree share model also enjoys this property;
see~\cite{dockins2009} for details.

On the other hand, splittability does not hold for heaps (for which the image of $\star$ is
empty), for named permissions (singletons cannot be split), or for counting
permissions (where $(v,-1)$ is not in the image of $\star$). Splittability also fails for
Petri nets, as the marking assigning one token to one place, with all other
places empty, cannot be split.
\end{example}

\subsection{Cross-split}

This more complicated property requires that if a heap can be split in
two different ways, then there should be intersections of these
splittings.  Formally, in a separation algebra $(H,\circ,\epsilon)$,
if $h_1\circ h_2 = h_0$ and $h_3\circ h_4 = h_0$, then there should be
four elements $h_{13},h_{14},h_{23},h_{24}$, informally representing
the intersections $h_1\cap h_3$, $h_1\cap h_4$, $h_2\cap h_3$ and
$h_2\cap h_4$ respectively, such that $h_{13}\circ h_{14} = h_1$,
$h_{23}\circ h_{24} = h_2$, $h_{13}\circ h_{23} = h_3$, and
$h_{14}\circ h_{24} = h_4$. The corresponding condition on Kripke
relational frames is obvious.  The following sound rule naturally
captures cross-split, where $p,q,s,t,u,v,x,y,z$ are labels,
and the labels $p,q,s,t$ do not occur in the conclusion:
\begin{center}
\AxiomC{$\myseq{\Gcal;(x,y\simp z);(u,v\simp z');(p,q\simp x);(p,s\simp u);(s,t\simp
    y);(q,t\simp v)}{\Gamma}{\Delta}$}
\AxiomC{$E(\Gcal)\vdash z = z'$}
\RightLabel{$CS$}
\BinaryInfC{$\myseq{\Gcal;(x,y\simp z);(u,v\simp z')}{\Gamma}{\Delta}$}
\DisplayProof
\end{center}

\begin{example}
All examples of separation algebras presented in
Section~\ref{subsec:pasl} satisfy $CS$; we are not aware of any
separation algebras with applications to program verification that fail to do so.
In the case of heaps, for example, the cross-splits are simply defined as
intersections, but the situation becomes more complex in the case that sharing
is possible, and in some cases the sub-splittings $h_{13},h_{23},\ldots$ need not
be uniquely defined.

For example, take the case of counting permissions. Let $h=\{l\mapsto(v,1)\}$
for some location $l$ and value $v$. We will abuse notation by writing this as
$h=1$, as the identity of the location and value are not important here. Let
$(h_1,h_2,h_3,h_4)$ be $(-2,3,-3,4)$ respectively. Then the values of
$(h_{13},h_{14},h_{23},h_{24})$ may be $(-2,\mbox{undefined},-1,4)$, or
$(\mbox{undefined},-2,-3,6)$, or $(-1,-1,2,5)$.

However, the definition of cross-split does not require uniqueness; it is sufficient to
describe a method by which a valid cross-split may be defined. We
consider each location $l\in dom(h)$ in turn. The most complex case has that $l$
is in the domain of all of $h_1,h_2,h_3,h_4$, and that the two splittings are not
identical on $l$. By definition at least one of $h_1(l),h_2(l)$ is negative, and
similarly $h_3(l),h_4(l)$. Without loss of generality say $h_1(l)$ is the
\emph{strictly largest} negative number of the four.
Then we may set $(h_{13}(l),h_{14}(l),h_{23}(l),h_{24}(l))$ to be
$(h_1(l),\mbox{undefined},h_3(l)-h_1(l),h_4(l))$. Routine calculation confirms
that this indeed defines a cross-split.
\end{example}

\subsection{Separation algebras without cancellativity}
\label{sec:nocancel}

We finally note that there has been interest in dropping the
cancellativity requirement from the definition of separation algebra~\cite{Gotsman:Precision}.
In our framework we need merely omit the $C$ rule of Figure~\ref{fig:sepalg}.

\begin{example}
\label{ex:noncancel}
The partial commutative semigroup construction of
Example~\ref{ex:partcomsemi}, as noted after that example, need not yield a
cancellative structure. In particular, if there exists an idempotent element
$v\star v=v$, then $\{l\mapsto v\}\circ\{l\mapsto v\}=\{l\mapsto v\}=
\{l\mapsto v\}\circ\epsilon$, but clearly $\{l\mapsto v\}\neq\epsilon$.
We give two examples:
\begin{itemize}
  \item
Monotonic Counters for Fictional Separation
Logic~\cite{JensenBirkedal2012}: fictional separation logic is a
program verification framework in which every module is associated
with its own notion of resource. We here note only the example of a
monotonic counter, for which the partial commutative semigroup is the
integers with a bottom element, with $max$ as the operation. Clearly
every element is idempotent.
  \item
Logical Heaps for Relaxed Separation Logic~\cite{Vafeiadis:Relaxed}: we refer
to \emph{op. cit.} for the rather involved definition,
 and an example of an idempotent element.
\end{itemize}
Both of the examples above satisfy indivisible unit, splittability, and cross-split, but
fail to satisfy disjointness.
\end{example}


%% file: equality_subs.tex
\section{From equality atoms to substitutions}
\label{sec:eq_subs}

Having equality atoms in the calculus is convenient when proving the
completeness of the calculus. However labelled calculi with
substitutions, cf.~\cite{Hou:Labelled15}, are easier to
implement. In particular, global substitution of labels reduces the
number of labels in the sequent, and simplifies the structure of the
sequent. This often leads to advantages in performance.

We begin by replacing the rules $id$, $\top^* L$, $\top^* R$, $\mand R$,
$\mimp L$, which in the proofs rules of Figure~\ref{fig:LS} contained equality atoms,
with versions without equality atoms, as shown in Figure~\ref{fig:lspslh_subs}.
These are
precisely the rules for BBI presented in~\cite{Hou:Labelled15}. We similarly transform the
rules $NEq$ and $EM$. Note the use of the explicit substitutions applied to the meta-variables $\Gcal$, $\Gamma$, and $\Delta$ in the rules $\top^* L$ and $EM$.

\begin{figure*}
\footnotesize
\centering
\begin{tabular}{cc}
\multicolumn{2}{l}{\textbf{Identity:}}\\[10px]
\multicolumn{2}{c}{
\AxiomC{$$}
\RightLabel{\tiny $id$}
\UnaryInfC{$\myseq{\Gcal}{\Gamma;w:p}{w:p;\Delta}$}
\DisplayProof
}\\[15px]
\multicolumn{2}{l}{\textbf{Logical Rules:}}\\[10px]
\AxiomC{$$}
\RightLabel{\tiny $\bot L$}
\UnaryInfC{$\myseq{\Gcal}{\Gamma; w:\bot}{\Delta}$}
\DisplayProof
$\qquad$
\AxiomC{$\myseq{\Gcal[\epsilon/w]}{\Gamma[\epsilon/w]}{\Delta[\epsilon/w]}$}
\RightLabel{\tiny $\top^* L$}
\UnaryInfC{$\myseq{\Gcal}{\Gamma;w:\top^*}{\Delta}$}
\DisplayProof
&
\AxiomC{$$}
\RightLabel{\tiny $\top R$}
\UnaryInfC{$\myseq{\Gcal}{\Gamma}{w:\top;\Delta}$}
\DisplayProof
$\qquad$
\AxiomC{$$}
\RightLabel{\tiny $\top^* R$}
\UnaryInfC{$\myseq{\Gcal}{\Gamma}{\epsilon:\top^*;\Delta}$}
\DisplayProof\\[15px]
\AxiomC{$\myseq{\Gcal}{\Gamma;w:A;w:B}{\Delta}$}
\RightLabel{\tiny $\land L$}
\UnaryInfC{$\myseq{\Gcal}{\Gamma;w:A\land B}{\Delta}$}
\DisplayProof
&
\AxiomC{$\myseq{\Gcal}{\Gamma}{w:A;\Delta}$}
\AxiomC{$\myseq{\Gcal}{\Gamma}{w:B;\Delta}$}
\RightLabel{\tiny $\land R$}
\BinaryInfC{$\myseq{\Gcal}{\Gamma}{w:A\land B;\Delta}$}
\DisplayProof\\[15px]
\AxiomC{$\myseq{\Gcal}{\Gamma}{w:A;\Delta}$}
\AxiomC{$\myseq{\Gcal}{\Gamma;w:B}{\Delta}$}
\RightLabel{\tiny $\limp L$}
\BinaryInfC{$\myseq{\Gcal}{\Gamma;w:A\limp B}{\Delta}$}
\DisplayProof
&
\AxiomC{$\myseq{\Gcal}{\Gamma;w:A}{w:B;\Delta}$}
\RightLabel{\tiny $\limp R$}
\UnaryInfC{$\myseq{\Gcal}{\Gamma}{w:A\limp B; \Delta}$}
\DisplayProof\\[15px]
\AxiomC{$\myseq{\Gcal;(x,y \simp z)}{\Gamma;x:A;y:B}{\Delta}$}
\RightLabel{\tiny $\mand L$}
\UnaryInfC{$\myseq{\Gcal}{\Gamma;z:A\mand B}{\Delta}$}
\DisplayProof
&
\AxiomC{$\myseq{\Gcal;(x,z \simp y)}{\Gamma;x:A}{y:B;\Delta}$}
\RightLabel{\tiny $\mimp R$}
\UnaryInfC{$\myseq{\Gcal}{\Gamma}{z:A\mimp B;\Delta}$}
\DisplayProof\\[15px]
\multicolumn{2}{c}{
\AxiomC{$\myseq{\Gcal;(x,y \simp z)}{\Gamma}{x:A;z:A\mand B;\Delta}$}
\AxiomC{$\myseq{\Gcal;(x,y \simp z)}{\Gamma}{y:B;z:A\mand B;\Delta}$}
\RightLabel{\tiny $\mand R$}
\BinaryInfC{$\myseq{\Gcal;(x,y \simp z)}{\Gamma}{z:A\mand B;\Delta}$}
\DisplayProof
}\\[15px]
\multicolumn{2}{c}{
\AxiomC{$\myseq{\Gcal;(x,y \simp z)}{\Gamma;y:A\mimp B}{x:A;\Delta}$}
\AxiomC{$\myseq{\Gcal;(x,y \simp z)}{\Gamma;y:A\mimp B; z:B}{\Delta}$}
\RightLabel{\tiny $\mimp L$}
\BinaryInfC{$\myseq{\Gcal;(x,y \simp z)}{\Gamma;y:A\mimp B}{\Delta}$}
\DisplayProof
}
\\[15px]
\multicolumn{2}{l}{\textbf{Structural Rules:}}\\[15px]
\AxiomC{ \ }
\RightLabel{\tiny $NEq$}
\UnaryInfC{$\myseq {\Gcal; w \not = w} \Gamma \Delta$}
\DisplayProof
&
\AxiomC{$\myseq{\Gcal[x/y]} {\Gamma[x/y]} \Delta[x/y]$}
\AxiomC{$\myseq{\Gcal; x \not = y} \Gamma \Delta$}
\RightLabel{\tiny $EM$}
\BinaryInfC{$\myseq \Gcal \Gamma \Delta$}
\DisplayProof
\\[15px]
\AxiomC{$\myseq{\Gcal[x/z];(x,\epsilon\simp x)}{\Gamma[x/z]}{\Delta[x/z]}$}
\RightLabel{\tiny $E_1$}
\UnaryInfC{$\myseq{\Gcal;(x,\epsilon\simp z)}{\Gamma}{\Delta}$}
\DisplayProof
&
\AxiomC{$\myseq{\Gcal[z/x];(z,\epsilon\simp z)}{\Gamma[z/x]}{\Delta[z/x]}$}
\RightLabel{\tiny $E_2$}
\UnaryInfC{$\myseq{\Gcal;(x,\epsilon\simp z)}{\Gamma}{\Delta}$}
\DisplayProof\\[15px]
\AxiomC{$\myseq{\Gcal;(x,\epsilon \simp x)}{\Gamma}{\Delta}$}
\RightLabel{\tiny $U$}
\UnaryInfC{$\myseq{\Gcal}{\Gamma}{\Delta}$}
\DisplayProof
&
\AxiomC{$\myseq{(y,x \simp z);(x,y \simp z);\Gcal}{\Gamma}{\Delta}$}
\RightLabel{\tiny $Com$}
\UnaryInfC{$\myseq{(x,y \simp z);\Gcal}{\Gamma}{\Delta}$}
\DisplayProof
\\[15px]
    \AxiomC{$\myseq {\Gcal;(u,y \simp x);(v,w \simp y);(z, w \simp x); (u, v \simp
        z)} \Gamma \Delta$}
    \RightLabel{\tiny $A$}
    \UnaryInfC{$\myseq {\Gcal;(u,y \simp x);(v,w \simp y)} \Gamma \Delta$}
\DisplayProof
&
\AxiomC{$\myseq{\Gcal[y/w];(x,y\simp z)}{\Gamma[y/w]}{\Delta[y/w]}$}
\RightLabel{\tiny $C$}
\UnaryInfC{$\myseq{\Gcal;(x,y\simp z);(x,w\simp z)}{\Gamma}{\Delta}$}
\DisplayProof
\\[15px]
\AxiomC{$\Gcal[y/z];\myseq{(w,x\simp y)}{\Gamma[y/z]}{\Delta[y/z]}$}
\RightLabel{\tiny $P$}
\UnaryInfC{$\myseq{\Gcal;(w,x\simp y);(w,x\simp z)}{\Gamma}{\Delta}$}
\DisplayProof
&
\AxiomC{$\myseq{\Gcal[\epsilon/x];(\epsilon,\epsilon\simp z)}{\Gamma[\epsilon/x]}{\Delta[\epsilon/x]}$}
\RightLabel{\tiny $D$}
\UnaryInfC{$\myseq{\Gcal;(x,x\simp z)}{\Gamma}{\Delta}$}
\DisplayProof
\\[15px]
\multicolumn{2}{l}{\textbf{Side conditions:}} \\
\multicolumn{2}{l}{In $\mand L$ and $\mimp R$, the labels $x$ and $y$
do not occur in the conclusion. }\\
\multicolumn{2}{l}{In the rule $A$, the label $z$ does not occur in the conclusion.}
\end{tabular}
\caption{The labelled sequent calculus $\lspslh'+D$ with explicit global label substitutions.} 
\label{fig:lspslh_subs}
\end{figure*}

If we are to move beyond the core sublogic $\ls$ of Figure~\ref{fig:LS} to get the
full logic of abstract separation algebras, we need a general method for converting
general structural rules synthesised from frame axioms, to proof rules without equality
atoms. Given a structural rule of the following form:
\begin{center}
\AxiomC{$\myseq {\Gcal; S_1; \dots; S_k; T_1; \dots; T_l} \Gamma \Delta$}
\AxiomC{$E(\Gcal) \vdash s_1 = t_1 \qquad \cdots \qquad E(\Gcal) \vdash s_p = t_p$}
\BinaryInfC{$\myseq \Gcal \Gamma \Delta$}
\DisplayProof
\end{center}
we use the following algorithm:
\begin{enumerate}
\item Delete each judgment $E(\Gcal) \vdash s_i = t_i$ and modify the proof rule
  as follows. If $t_i$ is not $\epsilon$ we replace it with $s_i$ everywhere it appears.
  Conversely if $t_i$ is $\epsilon$, but $s_i$ is not, replace $s_i$ with $\epsilon$
  everywhere.
\item For each $S_i$, if it is an equality $x = y$, we delete the equality and change the
  \emph{premise only} of the proof rule as follows. If $x$ is not $\epsilon$ we replace
  $x$ everywhere in the premise by $y$, and apply the explicit
  substitution $[y/x]$ to the occurrences of $\Gcal$,
  $\Gamma$, and $\Delta$ in the premises. We likewise produce another new proof rule in which $y$
  is replaced by $x$ everywhere in the premise, provided $y$ is not $\epsilon$.
\end{enumerate}

As an example, Figure~\ref{fig:lspslh_subs} presents the equality-free
rules for propositional abstract separation logic with
disjointness. We refer to this logic as $\psl_D$ and the proof system
as $\lspslh'+D$, using the prime to denote ``equality free''. We
choose to present this particular logic because it will be used for
the experiments of the next section. Note that the frame rule $E$
obtained from Formula~\ref{eq:ax} in Section~\ref{subsec:lspsl} is
split into two rules $E_1,E_2$, following step (2) of the procedure,
depending on the choice of substituting $x$ for $z$ in the premise, or
vice versa. We do not need to split $C$ and $P$ into two rules because
we define our sequents via sets, rather than lists, and so the two
resulting rules would be identical up to label renaming.


The soundness of the converted system is straightforward, since global
substitution of labels naturally captures the equality of the two
labels. 

Given a sequent $\myseq{\Gcal}{\Gamma}{\Delta}$ in any instance
calculus of $\lsg$, we convert this sequent to a sequent in the
equality-free version via a translation $\tau$, which simply performs
global substitutions to unify labels $h_1,h_2$, for each equality
relational atom $h_1 = h_2$ in $\Gcal$, with the restriction that we
cannot substitute for $\epsilon$. The corresponding equality-free
sequent is written as $\tau(\myseq{\Gcal}{\Gamma}{\Delta})$.

The following lemma is easy to show.

\begin{lemma}
\label{lem:unify_eq}
For any sequent $\myseq{\Gcal}{\Gamma}{\Delta}$, if $E(\Gcal)\vdash h
= h'$, then $h$ and $h'$ are unified in
$\tau(\myseq{\Gcal}{\Gamma}{\Delta})$.
\end{lemma}

We then show that the conversion preserves completeness. This can be
proved by a straightforward induction on the height of the derivation.

\begin{theorem}
For any sequent $\myseq{\Gcal}{\Gamma}{\Delta}$, if it is derivable in
an instance calculus $LS_i$ of $\lsg$, then
$\tau(\myseq{\Gcal}{\Gamma}{\Delta})$ is derivable in the
corresponding equality-free calculus $LS_i'$ given by the above
conversion.
\end{theorem}

\begin{corollary}
For any $\psl$ formula $F$, if it is derivable in an instance calculus
$LS_i$, then it is also derivable in the corresponding equality-free
calculus $LS_i'$ given by the above conversion.
\end{corollary}

It may seem odd that the equality judgements and their associated
rules can just be deleted completely when converting an equality-based
calculus into an equality-free calculus with global
substitutions. Intuitively, remember that we are intending for these
calculi to be used bottom-up: thus the global substitutions trivially
``implement'' the required equalities as backward proof-search
proceeds upwards towards the leaves.


%% file: examples.tex
\subsection{Example Derivations}

We now present some example derivations of valid formulae using
the calculus with label substitutions for the logic $\psl+D$. We shall use some inference rules for
derived logical connectives $\lnot$ and $\lor$, as shown below:

\begin{center}
\begin{tabular}{cc}
\AxiomC{$\myseq{\Gcal}{\Gamma}{w:A;\Delta}$}
\RightLabel{\tiny $\lnot L$}
\UnaryInfC{$\myseq{\Gcal}{\Gamma;w:\lnot A}{\Delta}$}
\DisplayProof
&
\AxiomC{$\myseq{\Gcal}{\Gamma;w:A}{\Delta}$}
\RightLabel{\tiny $\lnot R$}
\UnaryInfC{$\myseq{\Gcal}{\Gamma}{w:\lnot A;\Delta}$}
\DisplayProof\\[15px]
\AxiomC{$\myseq{\Gcal}{\Gamma;w:A}{\Delta}$}
\AxiomC{$\myseq{\Gcal}{\Gamma;w:B}{\Delta}$}
\RightLabel{\tiny $\lor L$}
\BinaryInfC{$\myseq{\Gcal}{\Gamma;w:A\lor B}{\Delta}$}
\DisplayProof
&
\AxiomC{$\myseq{\Gcal}{\Gamma}{w:A;w:B;\Delta}$}
\RightLabel{\tiny $\lor R$}
\UnaryInfC{$\myseq{\Gcal}{\Gamma}{w:A\lor B;\Delta}$}
\DisplayProof\\[15px]
\end{tabular}
\end{center} 

First let us consider the formula
$\lnot (\top^* \land A \land (B \mand \lnot (C \mimp (\top^* \limp
A))))$. This formula can be easily proved by the following
derivation:
\begin{center}
\AxiomC{}
\RightLabel{$id$}
\UnaryInfC{$(c,b\simp \epsilon);(a,b\simp \epsilon);\epsilon:A;a:B;c:C\vdash \epsilon:A$}
\RightLabel{$\top^* L$}
\UnaryInfC{$(c,b\simp d);(a,b\simp \epsilon);\epsilon:A;a:B;c:C;d:\top^*\vdash d:A$}
\RightLabel{$\limp R$}
\UnaryInfC{$(c,b\simp d);(a,b\simp \epsilon);\epsilon:A;a:B;c:C\vdash d:
\top^* \limp A$}
\RightLabel{$\mimp R$}
\UnaryInfC{$(a,b\simp \epsilon);\epsilon:A;a:B\vdash b:C \mimp
(\top^* \limp A)$}
\RightLabel{$\lnot L$}
\UnaryInfC{$(a,b\simp \epsilon);\epsilon:A;a:B; b:\lnot (C \mimp
(\top^* \limp A))\vdash$}
\RightLabel{$\mand L$}
\UnaryInfC{$\epsilon:A;\epsilon:B \mand \lnot (C \mimp
(\top^* \limp A))\vdash$}
\RightLabel{$\top^* L$}
\UnaryInfC{$w:\top^*;w:A;w:B \mand \lnot (C \mimp
(\top^* \limp A))\vdash$}
\RightLabel{$\land L$}
\UnaryInfC{$w:\top^* \land A \land (B \mand \lnot (C \mimp
(\top^* \limp A)))\vdash$}
\RightLabel{$\lnot R$}
\UnaryInfC{$\vdash w:\lnot (\top^* \land A \land (B \mand \lnot (C \mimp
(\top^* \limp A))))$}
\DisplayProof
\end{center} 

\noindent This formula is valid in non-deterministic BBI. However, we
are not aware of any proof for this formula in other proof systems
for non-deterministic BBI, such as the nested sequent
calculus~\cite{park2013} or the display
calculus~\cite{Brotherston10}. In particular, we ran this formula using the nested
sequent calculus based prover~\cite{park2013} for a week on a CORE i7
2600 processor, without success.

The formula $(F * F) \limp F$, where $F = \lnot(\top \mimp \lnot
\top^*)$, is not valid in non-deterministic BBI, but is valid in BBI
plus partial-determinism. To prove this formula, we use the following
derivation, where we write $\cdot^2$ for two consecutive applications
of a rule and use semi-colon to denote sequencing of rule
applications:

\begin{center}
\AxiomC{$(w',w''\simp \epsilon);(b',c'\simp w'');(b,c\simp w');(b,c\simp a);\cdots$}
\RightLabel{$A$}
\UnaryInfC{$(w',c'\simp w);(w,b'\simp \epsilon);\cdots$}
\RightLabel{$E^2$}
\UnaryInfC{$(c',w'\simp w);(b,c\simp w');(b',w\simp \epsilon);\cdots$}
\RightLabel{$A$}
\UnaryInfC{$(b',w\simp \epsilon);(\epsilon,b\simp w);(c',c\simp \epsilon);\cdots$}
\RightLabel{$A$}
\UnaryInfC{$(b,c\simp a);(b',b\simp \epsilon);(c',c\simp \epsilon);(\epsilon,\epsilon\simp \epsilon);\cdots$}
\RightLabel{$U$}
\UnaryInfC{$(b,c\simp a);(b',b\simp \epsilon);(c',c\simp \epsilon);a:\top\mimp\lnot\top^*;b':\top,c':\top\vdash$}
\RightLabel{$\top^* L^2$}
\UnaryInfC{$(b,c\simp a);(b',b\simp b'');(c',c\simp c'');a:\top\mimp\lnot\top^*;b':\top,c':\top;b":\top^*;c'':\top^*\vdash$}
\RightLabel{$\lnot R^2$}
\UnaryInfC{$(b,c\simp a);(b',b\simp b'');(c',c\simp c'');a:\top\mimp\lnot\top^*;b':\top,c':\top\vdash b'':\lnot\top^*;c'':\lnot\top^*$}
\RightLabel{$\mimp R^2$}
\UnaryInfC{$(b,c\simp a); a:\top\mimp\lnot\top^*\vdash b:\top\mimp\lnot\top^*;c:\top\mimp\lnot\top^*$}
\RightLabel{$\lnot L^2;\lnot R$}
\UnaryInfC{$(b,c\simp a);b:\lnot(\top\mimp\lnot\top^*);c:\lnot(\top\mimp\lnot\top^*)\vdash a:\lnot(\top\mimp\lnot\top^*)$}
\RightLabel{$\mand L$}
\UnaryInfC{$a:F\mand F\vdash a:F$}
\RightLabel{$\limp R$}
\UnaryInfC{$\vdash a:(F\mand F)\limp F$}
\DisplayProof
\end{center}

\noindent Some sequents in the above derivation are too long, so we omit the part of a sequent that is not important in the rule application. The top
sequent above the $A$ rule instance contains only one non-atomic
formula: $a:\top\mimp \lnot \top^*$ on the left hand side. The
correct relational atom that is required to split $a:\top\mimp \lnot
\top^*$ is $(w'',a\simp \epsilon)$. However, so far we have only
obtained $(w'',w'\simp \epsilon)$. Although $w'$ and $a$ both have
exactly the same children ($b$ and $c$), the non-deterministic monoid
allows the composition $b\circ c$ to yield multiple elements, or even
$\epsilon$. Thus we cannot conclude that $w' = a$ in
non-deterministic BBI. This can be solved by applying the rule $P$ to
replace $w'$ by $a$, then use $E$ to obtain $(w'',a\simp \epsilon)$
on the left hand side of the sequent, then the derivation can go
through:

\begin{center}
\AxiomC{}
\RightLabel{$\top^* R$}
\UnaryInfC{$(w'',a\simp \epsilon);\cdots;\vdash \epsilon:\top^*$}
\RightLabel{$\lnot L$}
\UnaryInfC{$(w'',a\simp \epsilon);\cdots;\epsilon:\lnot\top^*\vdash$}
\AxiomC{}
\RightLabel{$\top R$}
\UnaryInfC{$(w'',a\simp \epsilon);\cdots\vdash w'':\top$}
\RightLabel{$\mimp L$}
\BinaryInfC{$(w'',a\simp \epsilon);\cdots;a:\top\mimp\lnot\top^*;b':\top,c':\top\vdash$}
\DisplayProof
\end{center}

The proof system $\lspslh'+D$ does not have a rule for indivisible
unit. We show that we can use the rule $D$ for disjointness to prove
the axiom $\top^*\land (A\mand B)\limp A$ of indivisible unit.
We highlight the principal relational atoms where they are not obvious.
\begin{center}
\AxiomC{}
\RightLabel{$id$}
\UnaryInfC{$\myseq{(\epsilon,\epsilon\simp \epsilon);\cdots}{\epsilon:A; \epsilon: B}{\epsilon: A}$}
\RightLabel{$Eq_1$}
\UnaryInfC{$\myseq{(\epsilon,a_1\simp \epsilon);\cdots}{a_1:A; \epsilon: B}{\epsilon: A}$}
\RightLabel{ $E$}
%
\UnaryInfC{$\myseq{(a_1,w_2\simp w_1);$$(\epsilon,\epsilon\simp w_2);$\fbox{$(a_1,\epsilon\simp \epsilon)$};$\cdots}
{a_1:A; \epsilon: B}{\epsilon: A}$}
%
\RightLabel{ $D$}
\UnaryInfC{$\myseq{(a_1,w_2\simp w_1);$\fbox{$(a_2,a_2\simp w_2)$}$;(a_1,a_2\simp \epsilon);\cdots}
{a_1:A; a_2: B}{ \epsilon: A}$}
\RightLabel{$A$}
\UnaryInfC{$\myseq{(a_1,w_1\simp \epsilon);$\fbox{$(\epsilon,a_2\simp w_1);(a_1,a_2\simp \epsilon)$}$;\cdots}
{a_1:A; a_2: B}{\epsilon: A}$}
\RightLabel{$A$}
\UnaryInfC{$\myseq{(\epsilon,\epsilon\simp \epsilon);(a_1,a_2\simp \epsilon)}{a_1:A; a_2: B}{\epsilon: A}$}
\RightLabel{$U$}
\UnaryInfC{$\myseq{(a_1,a_2\simp \epsilon)}{a_1:A; a_2: B}{\epsilon: A}$}
\RightLabel{$\top^* L$}
\UnaryInfC{$\myseq{(a_1,a_2\simp a_0)}{a_0: \top^*; a_1:A; a_2: B}{a_0: A}$}
\RightLabel{$\mand L$}
\UnaryInfC{$\myseq{~}{a_0: \top^*; a_0:A\mand B}{ a_0: A}$}
\RightLabel{$\land L$}
\UnaryInfC{$\myseq{~}{a_0: \top^*\land (A\mand B)}{a_0: A}$}
\RightLabel{$\limp R$}
\UnaryInfC{$\myseq{~}{~}{a_0: \top^*\land (A\mand B)\limp A}$}
\DisplayProof
\end{center}


%% file: experiment.tex

\section{Implementation and Experiment}
\label{sec:experiment}

In this section we present an implementation, called Separata%
\footnote{Available at \url{http://users.cecs.anu.edu.au/~zhehou}.}%
, of a
semi-decision procedure for the logic $\psl_D$ (propositional abstract
separation logic with disjointness), and present experiments demonstrating its
efficacy. We focus on $\psl_D$ because, following Table~\ref{tab:models}, and
recalling that disjointness implies indivisible unit, it captures much of the
structure of the paradigmatic example of a separation
algebra, namely heaps. Heaps also satisfy the cross-split property, but we are
unaware of any theorems expressible in the language of abstract separation
logic that require
cross-split for proof (we are, however, aware of some such theorems when the
points-to connective $\mapsto$ is introduced). We note also that Separata may
handle various sublogics of $\psl_D$ got by variously omitting partiality,
cancellativity, or disjointness, or by weakening disjointness to indivisible unit.

Our implementation is based on the proof system $\lspslh'+D$ of
Figure~\ref{fig:lspslh_subs}, with three modifications. First,
we modify the rule $U$ to a rule $U'$ that creates the identity relational atom
$(x,\epsilon\simp x)$ only if $x$ occurs in the conclusion.
This does not reduce power:

\begin{lemma}
\label{lem:u_rule}
If $\myseq{\Gcal}{\Gamma}{\Delta}$ is derivable in $\lspslh'+D$, then it is derivable in the proof system with $U'$ replacing $U$.
\end{lemma}

Second, the rules $NEq$ and $EM$ are omitted, as they are not necessary for
completeness for $\psl_D$, nor its sublogics~\cite[Section 6.3]{houthesis}.

Third, although disjointness implies indivisible unit, we find that including
(the equality atom-free version of) the rule for indivisible unit can reduce the
number of labels in sequents and hence lead to a smaller search space and
better performance:
\begin{center}
  \AxiomC{$\myseq{\Gcal[\epsilon/x];(\epsilon,y\simp \epsilon)}{\Gamma[\epsilon/x]}{\Delta[\epsilon/x]}$}
  \RightLabel{IU}
  \UnaryInfC{$\myseq{\Gcal;(x,y\simp \epsilon)}{\Gamma}{\Delta}$}
  \DisplayProof
\end{center}

Our implementation is based on the following strategy when applying rules:
\begin{enumerate}
\item Try to close the branch by rules $id,\bot L,\top^* R,\top^* R$.
\item If (1) is not applicable, apply all possible $E_1,E_2,P,C,IU,D$
  rules to unify labels.
\item If (1-2) are not applicable, apply invertible rules $\land L$, $\land R$, $\limp L$, $\limp R$, $\mand L$, $\mimp R$, $\top^* L$ in all possible ways.
\item If (1-3) are not applicable, try $\mand R$ or $\mimp L$ by choosing existing relational atoms.
\item If (1-3) are not applicable, and (4) is not applicable because all combinations of $\mand R,\mimp L$ formulae and relational atoms are already applied, apply structural rules on the set $\Gcal_0$ of relational atoms in the sequent as follows.
\begin{enumerate}
\item Use $Com$ to generate all commutative variants of existing relational atoms in $\Gcal_0$, giving a set $\Gcal_1$.
\item Apply $A$ for each applicable pair in $\Gcal_1$, generating a set $\Gcal_2$. We do not apply $A$ to a pair $(u,y\simp x);(v,w\simp y)$ where
for any $z$ the pair $(z,w\simp x);(u,v\simp x)$ already exists in the sequent.
\item Use $U'$ to generate all identity relational atoms for each
  label in $\Gcal_2$, giving 
$\Gcal_3$.
\end{enumerate}
\item If (1-4) are not applicable, and (5) has been applied with result
$\Gcal_3=\Gcal_0$, then fail.
\end{enumerate}

Step (2) is terminating, because each substitution eliminates a label,
and we only have finitely many labels. It is also clear that step (5)
is terminating.

In the implementation, we view $\Gamma$ and $\Delta$ in a sequent $\myseq{\Gcal}{\Gamma}{\Delta}$
as lists,
and each time a logical rule is
applied, we place the subformulae at the front of the list. Thus our
proof search has a ``focusing flavour'', that always tries to
decompose the subformulae of a principal formula if possible. 
To guarantee completeness, each time we apply a $\mand R$ or $\mimp L$
rule, the principal formula is moved to the end of the list in the
corresponding premise, thus
each principal formula for these non-deterministic rules 
is
considered fairly, i.e., applied in turn.



\begin{table*}
\footnotesize
\centering
{
\begin{tabular}{|c|l|r|r|r|}
\hline
\vspace{-2mm}

& & & & \ \\
& Formula & BBeye  & $\fvlsbbi$ & \textbf{Separata} \\
   &     & (opt) & (heuristic) &  \\
\hline
\vspace{-2mm}
& & & & \ \\
(1) & $(a \mimp b) \land (\top \mand (\top^* \land a)) \limp b$ & 0.076 & 0.002 & 0.002\\
(2) & $(\top^* \mimp \lnot (\lnot a \mand \top^*)) \limp a$ & 0.080 & 0.004 & 0.002\\
(3) & $\lnot ((a \mimp \lnot (a \mand b)) \land ((\lnot a \mimp \lnot b) \land b))$ & 0.064 & 0.003 & 0.002\\
(4) & $\top^* \limp ((a \mimp (b \mimp c)) \mimp ((a \mand b) \mimp c))$ & 0.060 & 0.003 & 0.002\\
(5) & $\top^* \limp ((a \mand (b \mand c)) \mimp ((a \mand b) \mand c))$ & 0.071 & 0.002 & 0.004\\
(6) & $\top^* \limp ((a \mand ((b \mimp e) \mand c)) \mimp ((a \mand (b \mimp e)) \mand c))$ & 0.107 & 0.004 & 0.008\\
(7) & $\lnot ((a \mimp \lnot (\lnot (d \mimp \lnot (a \mand (c \mand b))) \mand a)) \land c \mand (d \land (a \mand b)))$ & 0.058 & 0.002 & 0.006\\
(8) & $\lnot ((c \mand (d \mand e)) \land B)$ where & 0.047 & 0.002 & 0.013\\
 & $ B := ((a \mimp \lnot (\lnot (b \mimp \lnot (d \mand (e \mand c)))
\mand a)) \mand (b \land (a \mand \top)))$ & & &\\
(9) & $\lnot ( C \mand (d \land (a \mand (b \mand e))))$ where & 94.230
& 0.003 & 0.053\\
 & $C := ((a \mimp \lnot (\lnot (d \mimp \lnot ((c \mand e) \mand (b
\mand a))) \mand a)) \land c)$ & & & \\
(10) & $(a \mand (b \mand (c \mand d))) \limp (d \mand (c \mand (b \mand a)))$ & 0.030 & 0.004 & 0.002\\
(11) & $(a \mand (b \mand (c \mand d))) \limp (d \mand (b \mand (c \mand a)))$ & 0.173 & 0.002 & 0.002\\
(12) & $(a \mand (b \mand (c \mand (d \mand e)))) \limp (e \mand (d \mand (a \mand (b \mand c))))$ & 1.810 & 0.003 & 0.002\\
(13) & $(a \mand (b \mand (c \mand (d \mand e)))) \limp (e \mand (b \mand (a \mand (c \mand d))))$ & 144.802 & 0.003 & 0.002\\
(14) & $\top^* \limp (a \mand ((b \mimp e) \mand (c \mand d)) \mimp ((a \mand d) \mand (c \mand (b \mimp e))))$ & 6.445 & 0.003 & 0.044\\
(15) & $\lnot(\top^* \land (a \land (b \mand \lnot(c \mimp (\top^* \limp a)))))$ & timeout & 0.003 & 0.003\\
(16) & $((D \limp (E \mimp (D \mand E))) \limp (b \mimp ((D \limp (E \mimp ((D \mand a) \mand a))) \mand b)))$, where  & 0.039 & 0.005 & 8.772\\
 & $D := \top^* \limp a$ and $E :=  a\mand a$ & & & \\
(17) & $((\top^* \limp (a \mimp (((a \mand (a \mimp b)) \mand \lnot b) \mimp (a \mand (a \mand ((a \mimp b) \mand \lnot b)))))) \limp$ & timeout & fail & 49.584\\
 & $((((\top^* \mand a) \mand (a \mand ((a \mimp b) \mand \lnot b))) \limp (((a \mand a) \mand (a \mimp b)) \mand \lnot b)) \mand \top^*))$ & & &\\ 
(18) & $(F\mand F)\limp F$, where $F := \lnot(\top \mimp \lnot \top^*)$ & invalid & invalid & 0.004\\
(19) & $(\top^* \land (a \mand b)) \limp a$ & invalid & invalid & 0.003\\
\hline
\end{tabular}
}
\caption{Experiment 1 results.}
\label{tab:experiment1}
\end{table*}

We incorporate a number of optimisations in proof search. (1)
Back-jumping~\cite{baader2003} is used to collect the ``unsatisfiable
core'' along each branch. When one premise of a binary rule has a
derivation, we try to derive the other premise only when the
unsatisfiable core is not included in that premise. 
(2) A search
strategy discussed by Park et al.~\cite{park2013} is also adopted. For
$\mand R$ and $\mimp L$ applications, we forbid the search to consider
applying the rule twice with the same pair of principal formula and
principal relational atom, since the effect is the same as
contraction, which is admissible. 
(3) Previous work on theorem proving for BBI has shown that associativity of $\mand$ is
a source of inefficiency in proof search~\cite{park2013,Hou:Labelled15}. We borrow
the idea of the heuristic method presented in~\cite{Hou:Labelled15} to
quickly solve certain associativity instances. When we detect
$z:A\mand B$ on the right hand side of a sequent, we try to search for
possible worlds (labels) for the subformulae of $A,B$ in the sequent,
and construct a binary tree using these labels. For example, if we can
find $x:A$ and $y:B$ in the sequent, we will take $x,y$ as the
children of $z$. When we can build such a binary tree of labels,
the corresponding relational atoms given by the binary tree will be
used (if they are in the sequent) as the prioritised ones when
decomposing $z:A\mand B$ and its subformulae. Without a
free-variable system, our handling of this heuristic method is just a
special case of the original one, but this approach can
speed up the search in certain cases.

The experiments in this paper are conducted on a Dell Optiplex 790 desktop with Intel CORE i7
2600 @ 3.4 GHz CPU and 8GB memory, running Ubuntu 13.04. The theorem provers are written in OCaml.

\paragraph{Experiment 1} We test our prover Separata for the logic $\psl_D$
on the formulae listed in Table~\ref{tab:experiment1}; the times displayed are in seconds and timeout is set as 1000s.
Where the formulae are valid in BBI, we compare
the results with two provers for BBI: the optimised implementation of
BBeye~\cite{park2013} and the incomplete heuristic-based $\fvlsbbi$~\cite{Hou:Labelled15}. We run BBeye in an iterative deepening way, and the time counted for BBeye is the total time it spends. Formulae (1-14) are used by Park et al. to test their prover BBeye for BBI~\cite{park2013}.
We can see that for formulae (1-14) the performance of Separata is comparable with the heuristic based prover for $\fvlsbbi$. Both provers are generally faster than BBeye. 
Formula (15) is one that BBeye had trouble with~\cite{Hou:Labelled15}, but Separata handles it trivially.
However, there are cases where BBeye is faster than Separata, for example, formula (16) found from a set of tests on randomly generated BBI theorems. 
Formula (17) is a converse example where a randomly
generated BBI theorem causes BBeye to time out and $\fvlsbbi$ with
heuristics to terminate within the timeout but without finding a proof
due to its incompleteness. 
Formula (18) is valid only when the monoid is partial, rather than merely
non-deterministic~\cite{wendling2010}, and formula (19) is the axiom of indivisible
unit. 

\paragraph{Experiment 2} 
Since $\psl_D$ is not finitely axiomatisable, we cannot enumerate all
theorems of this logic for testing purposes. In this experiment we instead use
randomly generated BBI theorems, which are a subset of all $\psl_D$
theorems, and so are provable in the calculus $\lspslh' +D$.
This test may not show the full power of our prover
Separata, but it allows comparison with existing provers. We generate
theorems via the Hilbert system for
BBI~\cite{GalmicheLarcheyWendling2006}, which consists of the axioms and
rules of classical logic plus those shown in Figure~\ref{fig:bbi_hilb}.

\begin{figure}[t!]
\begin{tabular}{ccc}
\textbf{Axioms} & \multicolumn{2}{c}{\textbf{Deduction Rules}}\\[5px]
\multirow{2}{*}{
\Axiom$\fCenter A \limp (\top^* \mand A)$
\alwaysNoLine
\UnaryInf$\fCenter(\top^* \mand A) \limp A$
\UnaryInf$\fCenter(A \mand B) \limp (B \mand A)$
\UnaryInf$\fCenter(A \mand (B \mand C)) \limp ((A \mand B) \mand C)$
\DisplayProof
}
&
\AxiomC{$\vdash A$}
\AxiomC{$\vdash A \limp B$}
\RightLabel{$MP$}
\BinaryInfC{$\vdash B$}
\DisplayProof
&
\AxiomC{$\vdash A \limp C$}
\AxiomC{$\vdash B \limp D$}
\RightLabel{$\mand$}
\BinaryInfC{$\vdash (A \mand B) \limp (C \mand D)$}
\DisplayProof\\[20px]
&
\AxiomC{$\vdash A \limp (B \mimp C)$}
\RightLabel{$\mimp 1$}
\UnaryInfC{$\vdash (A \mand B) \limp C$}
\DisplayProof
&
\AxiomC{$\vdash (A \mand B) \limp C$}
\RightLabel{$\mimp 2$}
\UnaryInfC{$\vdash A \limp (B \mimp C)$}
\DisplayProof
\end{tabular}\\[5px]
\caption{Some axioms and rules for the Hilbert system for BBI.}
\label{fig:bbi_hilb}
\end{figure}


We create random BBI theorems by first generating some random formulae
(not necessarily theorems) of length $n$, and globally substituting
these formulae at certain places in a randomly chosen BBI axiom
schema.
Then we use the deduction rules $\mimp 1$ and $\mimp 2$ in
Figure~\ref{fig:bbi_hilb}
to mutate the resultant formula at random places. The final theorem is
obtained by repeating the mutation for $i$ iterations.

The formulae generated from a fixed pair of $n$ and $i$ could have
different lengths since much randomness is involved in the generation,
but the average length grows as $n$ increases. BBI axioms do not
involve $\mimp$ at all, so the mutation step is vital to create
$\mimp$ in the theorems. In general, a higher $i$ value means that the
theorems are more structurally different from the BBI axioms. The
mutation step often makes formulae harder to prove.


\begin{table*}[t!]
\centering
{
\begin{tabular}{|c|c|c|c|c|c|c|}
\hline
Test & $n$ & $i$ & BBeye & BBeye & Separata & Separata\\
& & & proved & avg. time & proved & avg. time\\ 
\hline
1 & 10 & 20 & 74\% & 0.27s & 78\% & 0.19s \\
2 & 15 & 30 & 72\% & 4.63s & 66\% & 2.40s\\
3 & 20 & 30 & 61\% & 8.88s & 56\% & 8.76s\\
4 & 20 & 50 & 55\% & 12.39s & 53\% & 0.88s\\
5 & 30 & 50 & 49\% & 11.81s & 50\% & 6.26s\\
6 & 50 & 50 & 31\% & 11.82s & 31\% & 3.79s\\
\hline
\end{tabular}
}
\caption{Experiment 2 results.}
\label{tab:experiment3}
\end{table*}

We compare the performance of Separata with Park et al.'s BBI prover
BBeye against our randomly generated theorem suites in
Table~\ref{tab:experiment3}. The ``proved'' column for each prover
gives the percentage of successfully proved formulae within time out,
and ``avg. time'' column is the average time used when a formula is
proved. Timed out attempts are not counted in the average time. We
set the time out as 200 seconds. Formulae in test suite 1 have 20 to
30 binary connectives on average, while these numbers are around 100
to 150 for the formulae in test suite 6. Each test suit contains 100
BBI theorems. The two provers have similar successful rates in these
tests. Average time used on successful attempts increases for BBeye,
but fluctuates for Separata. In fact, both provers spent less than 1
second on most successful attempts (even for test suite 6), but there
were some ``difficult'' formulae that took them over 100 seconds to
prove. Therefore, we only include average time because median time
for both provers would be very small numbers that are difficult to
compare. In general there are fewer formulae that are ``difficult''
for Separata than for BBeye. But in test suites 2,3,4, Separata
proved fewer formulae in time. Lastly, we emphasise that Separata and
BBeye are designed for different logics, so comparing their
performance may not be fair; we do so only because there are no other
provers to compete with.


%% file: applications.tex
\section{Applications}
\label{sec:applications}

This section discusses applications that have been enabled by the
calculi of this paper. In particular, we discuss how
our calculi may be used in formal verification tasks, and how to extend
them with the widely-used points-to predicate for concrete models.

We have formalised the proof system $\lspslh'+D$ in
Isabelle/HOL~\cite{Separata-AFP}. As
an alternative to the OCaml implementation, we have developed
Isabelle tactics based on the proof search
techniques in this paper. The resultant proof method, also called
Separata, combines certain strengths of our calculus and Isabelle's
built-in proof methods. The Isabelle version of Separata serves as a
tool for machine-assisted formal verification, and it is compliant
with the existing separation algebra
library~\cite{Separation_Algebra-AFP}. We have also developed
advanced tactics for spatial connectives $\mand$ and $\mimp$ to
improve the performance and automation of the proof
method~\cite{hou2017}. As a result, our tactics can automatically
prove some lemmas in the seL4 development~\cite{Murray13} that
originally required several lines of manual proofs. H\'ou et al.~\cite{hou2017} also gives a case
study to show the potential of our tactics by proving
properties of the semantics of actions in local rely/guarantee
reasoning~\cite{Feng2009}.

Our unified framework for separation logic proof systems focuses on
abstract semantics, and it does not have the points-to predicate $\mapsto$,
which is essential to most real-life applications of separation logic. There are at
least two directions to use our framework
to solve this problem: (1) one can develop a proof system for
separation logic with a concrete semantics, such as the heap model,
by extending our labelled calculus with sound rules for the points-to
predicate; or (2)~ one can develop a new abstract separation logic
which extends this work with a ``points-to''-like predicate that is
defined in an abstract setting.

In separate work we have pursued the first direction to develop a
proof system for separation logic with the heap model~\cite{hou2015}.
The new labelled sequent calculus extends $\lspslh'+D$ with eight
inference rules for the points-to predicate in the heap model. Since
separation logic for the heap model is not recursively
enumerable~\cite{Calcagno2001,brochenin2012}, it is not possible to
obtain a sound, complete, and finite proof system for such a logic.
However, this incomplete proof system is still useful in the sense
that it can reason about data structures such as linked lists and
binary trees, and the eight rules for the points-to predicate are
powerful enough to be complete for the \emph{symbolic heap} fragment,
which is widely used in program verification. Furthermore, this proof
system can also prove many formulae that involve $\mimp$ and overlaid
data structures, such as formulae in the form of $(A \mand B) \land
(C \mand D)$ where $A$ and $C$ (similarly, $B$ and $D$) represent
overlaid resources. These cannot be expressed by symbolic heaps and
cannot be handled by most other proof methods.

In the other direction, we have developed a first-order abstract
separation logic which includes an ``abstract points-to''
predicate~\cite{hou2016}. The advantage of this abstract logic is
that it is recursively enumerable, and we have developed a sound and
complete proof system for it. Moreover, the abstract points-to
predicate can be equipped with various theories to approximate
concrete semantics of different flavours, such as Reynolds's
SL~\cite{reynolds2002}, Vafeiadis and Parkinson's
SL~\cite{vafeiadis2007}, Lee et al.'s SL~\cite{Lee2013}, and Thakur
et al.'s SL~\cite{Thakur2014}. Specifically, we can formulate the
properties of points-to in Reynolds's SL as a set $S$ of formulae
(theories) in the logic. When we prove a formula $F$ with such a
points-to predicate, we derive $S \vdash h:F$ in our proof system
where $h$ is an arbitrary world. If we want to prove $F$ in, e.g.,
Vafeiadis and Parkinson's SL, we merely need to modify the set $S$ to
some $S'$ that captures properties of points-to in the corresponding
semantics, and derive $S' \vdash h:F$. Again, it is impossible to
completely capture the points-to predicate in a concrete model, but
this abstract logic provides a language that is rich enough to
capture most of the inference rules in our work on the heap
model~\cite{hou2015}. The benefit of simulating various semantics via
theories is that we do not need to develop a new logic and prove
important properties for each kind of semantics, but merely need to
add or remove certain formulae from the set of theories.


%% file: rel_work.tex
\section{Related Work}
\label{sec:rel_work}

There are many more automated tools, formalisations, and logical
embeddings for separation logics than can reasonably be surveyed
within the scope of this paper. Almost all are not directly
comparable to this paper because (1) they focus on reasoning about
the specification language (Hoare triples) and they deal with
separation logic for some \emph{concrete}
semantics~\cite{berdine2005}; (2) they only consider a small subset
of the assertion language such as symbolic
heaps~\cite{Brotherston2014,Brotherston2016,Rowe2017}; or (3) they
focus on complexity and computability issues instead of automated
reasoning~\cite{Demri2016,Demri17}.

One exception to the above is Holfoot~\cite{Tuerk2009}, a HOL
mechanisation with support for automated reasoning about the `shape' of SL
specifications  -- exactly those aspects captured by abstract separation logic.
However, unlike Separata, Holfoot does not support magic wand. This is a common
restriction for automation of separation logic because magic wand is a source of
undecidability~\cite{brochenin2012}.
Conversely, the mechanisations and
embeddings that do incorporate magic wand tend to give little thought to (semi-)
decision procedures, for example,~\cite{simsthesis}. An exception to this are the tableaux
of~\cite{galmiche2010}, but their methods have not been implemented, and we
anticipate that such an implementation  would not be trivial.
Another partial exception to the trend to omit magic wand is
SmallfootRG~\cite{SmallfootRG}, which supports automation yet includes
\emph{septraction}~\cite{vafeiadis2007}, the De Morgan dual of magic wand.
However SmallfootRG does not support additive negation nor implication, and so
magic wand cannot be recovered; indeed in this setting septraction is mere syntactic
sugar that can be eliminated. We also note the
work~\cite{Schwerhoff:Lightweight,Blom:Witnessing} in which program proofs
involving magic wand can be automated, provided the code is annotated to assist
the prover. Recently, Reynolds et al.~\cite{Reynolds2016} gave a
decision procedure for quantifier-free heap model separation logic
which contains all the logical connectives in this work. They have
implemented an integrated subsolver in the DPLL-based SMT solver CVC4,
and their experiment has shown promising results. 
In contrast to this paper, which considers various abstract algebraic semantics and a unified proof theory for different models, Reynolds et al.'s work is focused on a concrete heap model semantics with the points-to predicate. 


Leaving out magic wand is not without cost, as the connective, while
surely less useful than $\mand$, has applications. A non-exhaustive list
follows: generating weakest preconditions via backwards
reasoning~\cite{IshtiaqOHearn01}; specifying
iterators~\cite{parkinson2005,neelakantan2006,haack09}; reasoning about
parallelism~\cite{Dodds2011}; and various applications of septraction, such as
the specification of iterators and buffers~\cite{Cherini09}. For a particularly
deeply developed example, see the correctness proof for the Schorr-Waite Graph
Marking Algorithm of~\cite{YangPhD}, which involves non-trivial inferences
involving magic wand (Lemmas 78 and 79). These examples provide ample motivation
to build proof calculi and tool support that include magic wand. Undecidability,
which in any case is pervasive in program proof, should not deter us from
seeking practically useful automation. 

This work builds upon earlier labelled sequent calculi for
BBI~\cite{Hou:Labelled15} and $\psl$~\cite{hou2013b}. The extensions
to previous work involve two main advances: first, a calculus framework
that handles any separation algebra property expressible in the general axiom
form; second, a new counter-model construction method that yields
completeness proofs for any calculus in our framework. The future work
section of the conference version of this paper~\cite{hou2013b} proposed some
putative rules extending this work to deal with Reynolds's heap semantics.
This proposal has since been developed in~\cite{hou2015}, in which we
presented the first theorem prover to handle all the logical connectives in a
separation logic with heap semantics. We believe that the theory
and techniques discussed in this paper can be extended to many
applications besides this model.

The link between BBI and separation logic is also emphasised as
motivation by Park et al~\cite{park2013}, whose BBI prover BBeye was used for
comparisons in Section~\ref{sec:experiment}.
Lee and Park~\cite{Lee2013} then extended~\cite{park2013}, 
independently to our own work, to a labelled sequent calculus for Reynolds's heap model 
with the restriction that all values are addresses.
Their paper~\cite{Lee2013} includes soundness and completeness
theorems, but unfortunately, our investigations~\cite{hou2015} showed that both claims are erroneous, and
they have since been retracted by the authors.

Also related, but so far not
implemented, are the tableaux for partial-deterministic (PD) BBI of Larchey-Wendling
and Galmiche~\cite{wendling2009,wendling2012}. In subsequent work, the
authors showed that
validity of partial-deterministic BBI-models coincides with validity of
cancellative partial-deterministic
BBI-models~\cite{larcheywendling2014}, which retrospectively shows
that their earlier work on PD-BBI was actually complete for $\psl$,
and that the rule $C$ of our system is admissible. We nevertheless
have presented $C$ in our system, partly to emphasise that we can
easily handle a range of structural properties in our system, and
partly because evidence in separation logic with concrete semantics,
and further
language constructors such as $\mapsto$ that may express properties of
the content of those semantics, show that the rule for cancellativity
may not always be admissible. For example, see the discussion
of~\cite{Gotsman:Precision}:
\begin{quote}
It is well-known that if $\mand$ is cancellative, then for a precise $q$ in $\mathcal{P}(\Sigma)$ and any $p_1,p_2$ in $\mathcal{P}(\Sigma)$, we have 
$(p_1\land p_2)\mand q = (p_1\mand q) \land (p_2\mand q)$
\end{quote}
where a predicate is ``precise'' if it ``unambiguously carves out an area of the heap''. It is not likely that our rule for cancellativity will remain admissible in
a semantics
in which such considerations are important.
It is not clear how Larchey-Wendling and Galmiche's tableaux method might be
extended to handle concrete semantics without an explicit rule for cancellativity. Moreover, their latest result does not include any treatment for non-deterministic BBI, nor for properties such as splittability and cross-split. In contrast, the relative ease
with which certain properties can be added or removed from labelled sequent
calculi is an important benefit of our approach.

Finally we note that the counter-model construction of this paper was necessary
to prove completeness because many of the properties we are interested in are
not BBI-axiomatisable, as proved by Brotherston and
Villard~\cite{brotherston2013}; that paper goes on to give a sound and complete
Hilbert axiomatisation of these properties by extending BBI  with techniques from hybrid logic.
Sequent calculus and proof search for this more powerful logic represents an
interesting future direction for research.
